\documentclass[preprint]{elsarticle}

\usepackage{template}
\usepackage{framed,multirow}

\usepackage{amssymb}
\usepackage{latexsym}
\usepackage{slashbox}

\usepackage{url}
\usepackage{color}
\usepackage{lineno}
\usepackage{amsmath,mathrsfs}
\usepackage{amsthm}
\usepackage[caption=false]{subfig}
\usepackage{dsfont}
\usepackage{float}
\usepackage{gensymb}
\usepackage{natbib}
\usepackage{graphicx,titlesec}
\usepackage{epstopdf}
\usepackage{geometry}
\usepackage{booktabs}
\definecolor{newcolor}{rgb}{.8,.349,.1}

\usepackage{varioref}
\usepackage{hyperref}
\usepackage{cleveref}

\hypersetup{
	colorlinks=true,
	linkcolor={red!50!black},
	citecolor={blue!50!black},
	urlcolor={blue!80!black}
}

\newcommand{\tabincell}[2]{\begin{tabular}{@{}#1@{}}#2\end{tabular}}

\def\m{\mbox{\boldmath $m$}}

\newcommand{\abs}[1]{\left\lvert#1\right\rvert}

\newtheorem{example}{Example}[section]
\newtheorem{lemma}{Lemma}

\crefformat{equation}{(#2#1#3)}
\crefmultiformat{equation}{(#2#1#3)}{ and~(#2#1#3)}{, (#2#1#3)}{ and~(#2#1#3)}

\crefformat{figure}{Fig.~#2#1#3}
\crefmultiformat{figure}{Figures~(#2#1#3)}{ and~(#2#1#3)}{, (#2#1#3)}{ and~(#2#1#3)}
\crefformat{table}{Table~#2#1#3}

\begin{document}


\begin{frontmatter}
\title{Numerical methods for antiferromagnetics}

\author[1]{Panchi Li}
\ead{LiPanchi1994@163.com}
\author[1,2]{Jingrun Chen\corref{cor1}}
\cortext[cor1]{Corresponding authors}
\ead{jingrunchen@suda.edu.cn}
\author[1,2]{Rui Du\corref{cor1}}
\ead{durui@suda.edu.cn}
\author[3]{Xiao-Ping Wang\corref{cor1}}
\ead{mawang@ust.hk}

\address[1]{School of Mathematical Sciences, Soochow University, Suzhou, 215006, China.}
\address[2]{Mathematical Center for Interdisciplinary Research, Soochow University, Suzhou, 215006, China.}
\address[3]{Department of Mathematics, The Hong Kong University of Science and Technology, Clear Water Bay, Kowloon, Hong Kong, China.}

\begin{abstract}
Compared with ferromagneitc counterparts, antiferromagnetic materials are considered as the future of spintronic applications
since these materials are robust against the magnetic perturbation, produce no stray field, and display ultrafast dynamics.
There are (at least) two sets of magnetic moments in antiferromagnets (with magnetization of the same magnitude but antiparallel directions)
and ferrimagnets (with magnetization of the different magnitude). The coupled dynamics for the bipartite collinear antiferromagnets
is modeled by a coupled system of Landau-Lifshitz-Gilbert equations with an additional term originated from the antiferromagnetic exchange, which
leads to femtosecond magnetization dynamics. In this paper, we develop three Gauss-Seidel projection methods for micromagnetics simulation in antiferromagnets and ferrimagnets. They are first-order accurate in time and second-order in space, and only solve linear systems of equations with constant coefficients at each step. Femtosecond dynamics, N\'{e}el wall structure, and phase transition in presence of an external magnetic field for antiferromagnets are
provided with the femtosecond stepsize.
\end{abstract}


\end{frontmatter}


\section{Introduction}

An electron has both charge and spin properties. The active manipulation of spin degrees of freedom in solid-state systems is known as spintronics \cite{ZuticFabianDasSarma:2004, Gomonay2014Review}. Most of researches have been focused on ferromagnets (FMs)
for GHz-frequency magnetization dynamics in the past decades and the domain wall velocity can reach $\sim100\;\textrm{m}/\textrm{s}$.
In early days, antiferromagnets (AFMs), however, are thought to be less effective for spintronic manipulations since they are robust against
the magnetic perturbation and produce no stray field \cite{Baltz2018Review}. In fact, it was latter realized that the robustness of AFMs
with respect to the magnetic perturbation makes them better candidates for spintronic applications due to the high stability of domain
wall structures. In addition, one striking feature in AFMs is the femtosecond magnetization dynamics due to the antiferromagnetic
exchange coupling, which has been employed to generate THz-frequency magnetization dynamics \cite{Puliafito2019Model,Ivanov2014femtosecond,Luis2019DW},
and boosts the domain wall velocity in AFMs to
$\sim10,000\;\textrm{m}/\textrm{s}$ \cite{Gomonay2016domainwall,Gomonay2016Highdominwall,Shiino2016domainwall}.
Moreover, room-temperature antiferromagnetic order has been found over a broad range of materials, such as metal, semiconductor, and insulator,
which can be used in spintronics devices such as racetrack memories, memristors, and sensors \cite{Luis2019DW, Baltz2018Review,Wadley2016Electrical,ultrafast2008Fiebig}.

From the modeling perspective, magnetization $\mathbf{M}$ is the basic quantity of interest. In ferromagnetic materials, its dynamics
is modeled by the Landau-Lifshitz-Gilbert (LLG) equation \cite{LandauLifshitz1935,Gilbert1955} with the property that $\abs{\mathbf{M}} = M_s$
where $M_s$ is the saturation magnetization. In general, an AFM contains $n$ magnetic sublattices and the magnetization
$\mathbf{M}=\sum_{\lambda=1}^n\mathbf{M}_{\lambda}=\mathbf{0}$ with $\mathbf{M}_{\lambda}\neq\mathbf{0}$.
Most common AFMs, such as FeMn and NiO,  have two sublattices $A$ and $B$, associated with two magnetization fields
$\mathbf{M}_A$ and $\mathbf{M}_B$ satisfying $\mathbf{M} = \mathbf{M}_A + \mathbf{M}_B = \mathbf{0}$. For each sublattice,
its magnetization satisfies a LLG equation with an antiferromagnetic exchange term which couples these two equations.
In the physics community, by introducing another order parameter $\mathbf{L} = \mathbf{M}_A - \mathbf{M}_B$,
reduced models for AFMs are developed for describing magnetization dynamics \cite{Ivanov2014femtosecond, Baltz2018Review}.

In this work, we focus on numerical methods for AFMs and ferrimagnets described by the coupled system of LLG equations.
By ferrimagnets such as Fe$_3$O$_4$, we mean $\abs{\mathbf{M}_A} \neq \abs{\mathbf{M}_B}$. A large volume of methods
have been proposed for FMs; see \cite{Prohl2006, CJreview2007, Cimrak2007} for reviews and references therein.
For the temporal discretization, there are explicit schemes\cite{alouges2006convergence,FourRK2008}, implicit schemes\cite{Yamada2004Implicit,bartels2006convergence,implicit2012}, and semi-implicit schemes\cite{NumGSPM2001,panchi2019GSPM,NumMethods2000,SecSemi2019,seim2005,Changjian2019semiimplicit}.
However, there is no work on numerical methods for AFMs in the literature.

At a first glance, numerical methods for FMs can be directly applied to AFMs and ferrimagnets with minor modifications.
However, pros and cons of different methods for FMs may not be directly transferred. For example, explicit methods
require sub-picosecond stepsize in micromagnetics simulation. For AFMs or ferrimagnets, the stepsize becomes
sub-femtosecond for explicit methods while the time scale of interest is of nanoseconds.
The underlying reason is that the antiferromagnetic exchange term poses an characteristic time scale of femtoseconds
on magnetization dynamics. Even unconditionally stable implicit and semi-implicit schemes have to resolve the magnetization
dynamics at femtosecond scales in order to capture the correct physics. Implicit schemes solve a nonlinear system
of equations at each step. From FMs to AFMs, the dimension of the nonlinear system is doubled with possibly more
solutions (locally stable magnetic structures). In semi-implicit schemes, the nonlinear structure of the coupled
system does not bring any difficulty in an explicit way and only the computational complexity is doubled.
Therefore, semi-implicit methods provide the best compromise between
stability and efficiency. Gauss-seidel projection methods (GSPMs) \cite{NumGSPM2001, panchi2019GSPM} are of our
first choice since only linear systems of equations with constant coefficients needs to be solved at each step.
As shown in the paper, the coupled system of LLG equations can be solved by GSPMs with the computational complexity
doubled at each step. Due to the antiferromagnetic exchange, the stepsize in GSPMs is femtosecond as expected.

The rest of this paper is organized as follows. We introduce the model for AFMs and ferrimagnets in \Cref{sec:model}.
The corresponding Gauss-Seidel projection methods are described in \Cref{sec:schemes}.
Their accuracy with respect to temporal and spatial stepsizes is validated in \Cref{sec:experiments}.
Femtosecond magnetization dynamics, N\'{e}el wall structures, and phase transition in presence of a magnetic field
are simulated for AFMs in \Cref{sec:simulaions}. Conclusions are drawn in \Cref{sec:conclusion}.

\section{Model for antiferromagnets: A coupled system of Landau-Lifshitz-Gilbert equations}\label{sec:model}

Consider a bipartite collinear magnetic material occupied by $\Omega\in\mathds{R}^3$. Its magnetization structure can be a ferromagnetic phase
(\cref{fig:magnetictypes}(a)), a ferrimagnetic phase (\cref{fig:magnetictypes}(b)), or an antiferromagnetic phase (\cref{fig:magnetictypes}(c)).
The dyanmics of the two-sublattices system is modeled by a coupled system of phenomenological LLG equations \cite{Baltz2018Review}
\begin{equation}
\left\{
\begin{aligned}
\partial_t\mathbf{M}_A &= -\mu_0\gamma\mathbf{M}_A\times\mathbf{H}_A + \frac{\alpha}{M_s}\mathbf{M}_A\times\partial_t\mathbf{M}_A, \\
\partial_t\mathbf{M}_B &= -\mu_0\gamma\mathbf{M}_B\times\mathbf{H}_B + \frac{\alpha}{M_s}\mathbf{M}_B\times\partial_t\mathbf{M}_B.
\end{aligned}
\right.
\label{LLG}
\end{equation}
\begin{figure}[htbp]
	\centering
	\subfloat[Ferromagnetism]{\label{fig:ferromagetic}\includegraphics[width=1.5in]{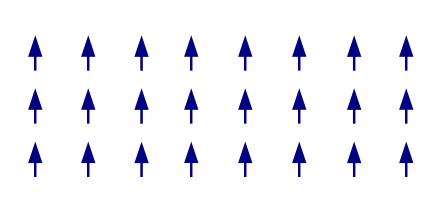}}
	\subfloat[Ferrimagnetism]{\label{fig:ferrimagetic}\includegraphics[width=1.5in]{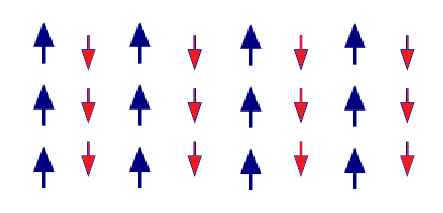}}
	\subfloat[Antiferromagnetism]{\label{fig:antiferromagetic}\includegraphics[width=1.5in]{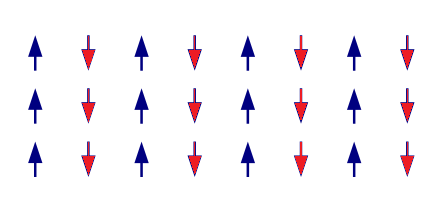}}
	\caption{Orientations of magnetic moments in magnetic materials in the ground state. (a) Ferromagneitc phase; (b) Ferrimagnetic phase; (c) Antiferromagnetic phase.}\label{fig:magnetictypes}
\end{figure}
$\alpha$ is the dimensionless damping parameter, $\gamma$ is the gyromagnetic ratio, and $\mu_0$ is the magnetic permeability of vacuum.
For each sublattice $\lambda = A, B$, we have magnetization $\mathbf{M}_{\lambda} = (M_{\lambda 1}, M_{\lambda 2}, M_{\lambda 3})^T$
with $|\mathbf{M}_{\lambda}| = M_s$. For ferrimagnets, $M_s$ is different for $A$ and $B$.
$\mathbf{H}_{\lambda} = -\delta W/\delta\mathbf{M}_{\lambda}$ with $W$ being the magnetic energy density of an antiferromagnetic or ferrimagnetic system including magnetic anisotropy, ferromagnetic exchange, antiferromagnetic exchange,
and Zeeman energy (external magnetic field) \cite{Puliafito2019Model,Luis2019DW}
\begin{equation}
\label{equ:energy}
W[\mathbf{M}_A, \mathbf{M}_B] = W_\textrm{a}[\mathbf{M}_A, \mathbf{M}_B] + W_\textrm{e}[\mathbf{M}_A, \mathbf{M}_B] +
W_\textrm{ae}[\mathbf{M}_A, \mathbf{M}_B] + W_\textrm{ext}[\mathbf{M}_A, \mathbf{M}_B].
\end{equation}
\eqref{LLG} can be viewed as two sets of LLG equations for $\mathbf{M}_A$ and $\mathbf{M}_B$ with the antiferromagnetic exchange coupling in \eqref{equ:energy}.

Details of the four terms in \eqref{equ:energy} are described as follows:
\begin{itemize}[leftmargin=*]
\item[1. ] \textbf{Anisotropy energy:} Magnetization usually favors an easy-axis direction of the form
\begin{equation*}
\label{Anisotropy-energy}
W_\textrm{a}[\mathbf{M}_A, \mathbf{M}_B] = \int_{\Omega}\Phi(\mathbf{M}_A) + \Phi(\mathbf{M}_B),
\end{equation*}
where $\Phi:\Omega \to \mathbb{R}^+$ is a smooth function. Suppose the easy-axis direction is the x-axis for a uniaxial material,
the total anisotropy energy of two sublattices is
\begin{equation*}
W_\textrm{a}[\mathbf{M}_A, \mathbf{M}_B] = \frac{K_u}{M_s^2}\int_{\Omega}(M_{A2}^2 + M_{A3}^2 + M_{B2}^2 + M_{B3}^2)\mathrm{d}\mathbf{x},
\label{equ:anisotropyenergy}
\end{equation*}
where $K_u$ is the material parameter.
\item[2. ]\textbf{Ferromagnetic exchange:} Magnetization of each sublattice experiences a ferromagnetic exchange energy of the form
\begin{equation*}
\label{Ferromagnetic-exchange}
W_\textrm{e}[\mathbf{M}_A, \mathbf{M}_B] = \frac{A}{M_s^2}\int_{\Omega}|\nabla\mathbf{M}_A|^2 + |\nabla\mathbf{M}_B|^2,
\end{equation*}
where $A$ presents the exchange constant of the material.
\item[3. ]\textbf{Antiferromagnetic exchange:} Magnetization of sublattice A and sublattice B
favors alignment along an antiparallel direction, thus the exchange energy is of the form
\begin{equation}
\label{AFM-exchange}
W_\textrm{ae}[\mathbf{M}_A, \mathbf{M}_B] = \frac{4A_{AFM}}{a^2M_s^2}\int_{\Omega}\mathbf{M}_A\cdot\mathbf{M}_B,
\end{equation}
where $A_{AFM}$ is the antiferromagnet exchange parameter and $a$ is the atomic lattice constant. For a positive $A_{AFM}$, the system
favors a ferromagnetic state. For a negative $A_{AFM}$, however, an antiferromagnetic state is preferred.
\item[4. ]\textbf{Zeeman energy:} In the presence of an external magnetic filed $\mathbf{H}_\textrm{ext}$, the interaction
energy is of the form
\begin{equation*}
\label{Zeeman-energy}
W_\textrm{ext}[\mathbf{M}_A, \mathbf{M}_B] = -\mu_0\int_{\Omega}\mathbf{H}_\textrm{ext}\cdot(\mathbf{M}_A + \mathbf{M}_B).
\end{equation*}
\end{itemize}
Thus, the free energy of an antiferromagnetic or ferrimagnetic material is explicitly written as
\begin{align}
\label{Free-energy}
W[\mathbf{M}_A, \mathbf{M}_B] = &\int_{\Omega}\Phi(\mathbf{M}_A) + \Phi(\mathbf{M}_B) +
\frac{A}{M_s^2}\int_{\Omega}|\nabla\mathbf{M}_A|^2 + |\nabla\mathbf{M}_B|^2 \nonumber\\
&+ \frac{4A_{AFM}}{a^2M_s^2}\int_{\Omega}(\mathbf{M}_A\cdot\mathbf{M}_B) -
\mu_0\int_{\Omega}\mathbf{H}_\textrm{ext}\cdot(\mathbf{M}_A + \mathbf{M}_B).
\end{align}

The system of LLG equations \eqref{LLG} can be rewritten equivalently as
\begin{equation}
\left\{
\begin{aligned}
  \frac{\partial\mathbf{M}_A}{\partial t} &= -\frac{\mu_0\gamma}{1+\alpha^2}\mathbf{M}_A\times\mathbf{H}_A
-\frac{\alpha\mu_0\gamma}{(1+\alpha^2)M_s}\mathbf{M}_A\times(\mathbf{M}_A\times\mathbf{H}_A), \\
\frac{\partial\mathbf{M}_B}{\partial t} &= -\frac{\mu_0\gamma}{1+\alpha^2}\mathbf{M}_B\times\mathbf{H}_B
-\frac{\alpha\mu_0\gamma}{(1+\alpha^2)M_s}\mathbf{M}_B\times(\mathbf{M}_B\times\mathbf{H}_B),
\end{aligned}
\right.
\label{LLG2}
\end{equation}
and the effective fields are
\begin{align*}
  \mathbf{H}_A &= -\frac{2K_u}{M_s^2}(M_{A2}\mathbf{e}_2+M_{A3}\mathbf{e}_3) + \frac{2A}{M_s^2}\Delta\mathbf{M}_A - \frac{4A_{AFM}}{a^2M_s^2}\mathbf{M}_B + \mu_0\mathbf{H}_\textrm{ext} \\
  \mathbf{H}_B &= -\frac{2K_u}{M_s^2}(M_{B2}\mathbf{e}_2+M_{B3}\mathbf{e}_3) + \frac{2A}{M_s^2}\Delta\mathbf{M}_B - \frac{4A_{AFM}}{a^2M_s^2}\mathbf{M}_A + \mu_0\mathbf{H}_\textrm{ext}
\end{align*}
with $\mathbf{e}_2 = (0,1,0)$ and $\mathbf{e}_3 = (0,0,1)$.

To ease the description, we now nondimensionlize \eqref{LLG2}.
Defining $\mathbf{M}_{\lambda} = M_s\mathbf{m}_{\lambda}$, $\mathbf{H}_\textrm{ext} = M_s\mathbf{h}_e$, $\mathbf{x} = L\mathbf{x'}$ with
$L$ the diameter of $\Omega$, and $W[\mathbf{M}_A, \mathbf{M}_B] = (\mu_0M_s^2)W'[\mathbf{m}_A, \mathbf{m}_B]$, we have
\begin{align}
  \label{equ:energy_dimensionless}
  W'[\mathbf{m}_A, \mathbf{m}_B] = &q\int_{\Omega'}(m_{A2}^2+m_{A3}^2 + m_{B2}^2+m_{B3}^2) +\epsilon\int_{\Omega'}|\nabla\mathbf{m}_A|^2 + |\nabla\mathbf{m}_B|^2 \nonumber\\
  &+ \delta\int_{\Omega'}(\mathbf{m}_A\cdot\mathbf{m}_B) - \int_{\Omega'}\mathbf{h}_{e}\cdot(\mathbf{m}_A + \mathbf{m}_B),
\end{align}
where $q = 2K_u/(\mu_0M_s^2)$, $\epsilon = 2A/(\mu_0M_s^2L^2)$, and $\delta = 4A_{AFM}/(\mu_0a^2M_s^2)$. Upon rescaling time $t\rightarrow (1+\alpha^2)(\mu_0\gamma M_s)^{-1}t$, \eqref{LLG2} can be rewritten as
\begin{equation}
\left\{
\begin{aligned}
  \frac{\partial\mathbf{m}_A}{\partial t} &= -\mathbf{m}_A\times\mathbf{h}_A
-\alpha\mathbf{m}_A\times(\mathbf{m}_A\times\mathbf{h}_A),  \\
\frac{\partial\mathbf{m}_B}{\partial t} &= -\mathbf{m}_B\times\mathbf{h}_B
-\alpha\mathbf{m}_B\times(\mathbf{m}_B\times\mathbf{h}_B),
\end{aligned}
\right.  \label{equ:LLG2_dimensionless}
\end{equation}
where
\begin{align}
  \mathbf{h}_A &= -q(m_{A2}\mathbf{e}_2 + m_{A3}\mathbf{e}_3) + \epsilon\Delta\mathbf{m}_A - \delta\mathbf{m}_B + \mathbf{h}_e,  \label{equ:effectiveA}\\
  \mathbf{h}_B &= -q(m_{B2}\mathbf{e}_2 + m_{B3}\mathbf{e}_3) + \epsilon\Delta\mathbf{m}_B - \delta\mathbf{m}_A + \mathbf{h}_e.
  \label{equ:effectiveB}
\end{align}
Homogeneous Neumann boundary conditions are used
\begin{equation}
\frac {\partial\mathbf{m}_A}{\partial\nu}|\Gamma = 0, \quad \frac {\partial\mathbf{m}_B}{\partial\nu}|\Gamma = 0,
\end{equation}
where $\Gamma = \partial\Omega$ and $\nu$ is the unit outward normal vector along $\Gamma$. It's worth mentioning that the above model
is also used for ferrimagnetic materials with one of magnetization, $\abs{\mathbf{m}_B} < 1$.

It is easy to check from \eqref{equ:LLG2_dimensionless} that the following statement is true.
\begin{lemma}
For $\lambda = A, B$, we have
\[
\abs{\mathbf{m}_{\lambda}(t,\mathbf{x})} =  \abs{\mathbf{m}_{\lambda}(t_0,\mathbf{x})}, \;\forall \mathbf{x}\in\Omega,\;\forall t>t_0.
\]
\end{lemma}

The antiferromagnetic exchange \eqref{AFM-exchange} plays an important role in AFMs. Consider the case when $\delta\rightarrow\infty$
and the system \eqref{equ:LLG2_dimensionless} reduces to
\begin{equation*}
\left\{
\begin{aligned}
  \frac{\partial\mathbf{m}_A}{\partial t} &= \delta\mathbf{m}_A\times\mathbf{m}_B
+\alpha\delta\mathbf{m}_A\times(\mathbf{m}_A\times\mathbf{m}_B), \\
\frac{\partial\mathbf{m}_B}{\partial t} &= \delta\mathbf{m}_B\times\mathbf{m}_A
+\alpha\delta\mathbf{m}_B\times(\mathbf{m}_B\times\mathbf{m}_A).
\end{aligned}
\right.
\end{equation*}
Combining the above two equations with $\mathbf{m} = (\mathbf{m}_A + \mathbf{m}_B)/2$, we get an equation of Bernoulli type
\begin{equation}
  \label{equ:ODE}
  \frac{\partial\mathbf{m}}{\partial t} = \alpha\delta\left(2|\mathbf{m}|^2-1\right)\mathbf{m} - \alpha\delta\mathbf{m}.
\end{equation}

\begin{lemma}\label{lem:exchange}
  $\forall \mathbf{x}\in\Omega$, if the antiferromagnetic exchange parameter $\delta>0$, the system favors the antiferromagnetic state, and if $\delta<0$,
  the system favors the ferromagnetic state.
\end{lemma}
\begin{proof}
The analytic solution of \eqref{equ:ODE} is
\[
|\mathbf{m}|^2 = \frac 1{1 - C\exp(4\alpha\delta t)}
\]
with $C$ a positive constant determined by the initial condition.

Therefore, when $\delta>0$, $t\rightarrow\infty$, we have $|\mathbf{m}|^2\rightarrow 0$, $\forall \mathbf{x}\in\Omega$.
When $\delta<0$, $t\rightarrow\infty$, we have $|\mathbf{m}|^2\rightarrow 1$, $\forall \mathbf{x}\in\Omega$.

The definition of $\mathbf{m}$ yields
\begin{align*}
|\mathbf{m}|^2 &= \left|\frac {\mathbf{m}_A+\mathbf{m}_B}{2}\right|^2 \\
&= \frac1{4}\left(|\mathbf{m}_A|^2 + |\mathbf{m}_B|^2 + 2(\mathbf{m}_A, \mathbf{m}_B)\right)\\
&= \frac 1{2}(1 + (\mathbf{m}_A, \mathbf{m}_B))\\
&= \frac 1{2}(1 + \cos(\mathbf{m}_A, \mathbf{m}_B))
\end{align*}

As a consequence, $\mathbf{m}_A =  -\mathbf{m}_B$ and $\mathbf{m}_A =  \mathbf{m}_B$ for $\delta>0$ and $\delta<0$, respectively.
\end{proof}

\Cref{lem:exchange} implies that the coupled system convergences to an antiferromagnetic state exponentially fast with the exponent proportational to the
antiferromagnetic exchange parameter $\delta$. This is indicated numerically by the energy decay in \Cref{energy}.


\section{Gauss-Seidel projection methods for antiferromagnetics}\label{sec:schemes}

In this section, we introduce three Gauss-Seidel projection methods for \eqref{equ:LLG2_dimensionless}.
The finite difference method is used for spatial discretization with unknowns
$\mathbf{m}_{\lambda}(i) = \mathbf{m}_{\lambda}((i-\frac 1{2})\Delta x)$ in 1D and
$\mathbf{m}_{\lambda}(i,j,k) = \mathbf{m}_{\lambda}((i-\frac 1{2})\Delta x, (j-\frac 1{2})\Delta y, (k-\frac 1{2})\Delta z)$ in 3D,
where $i=0,1,\cdots, M, M+1$, $j = 0,1,\cdots, N, N+1$, and $k = 0,1,\cdots, K, K+1$ and $M, N, K$ represent the number of segments
for each direction.

\subsection{Original Gauss-Seidel Projection Method}\label{subsec:GSPM}

This is a direct generalization of the original GSPM \cite{NumGSPM2001} to the antiferromagnetic or ferrimagnetic case.
For \eqref{equ:LLG2_dimensionless}, the GSPM works as follows.
Define the vector field for the splitting procedure:
\begin{align*}
  \mathbf{h}_A = \epsilon\Delta\mathbf{m}_A + \mathbf{\hat f}_A,\\
  \mathbf{h}_B = \epsilon\Delta\mathbf{m}_B + \mathbf{\hat f}_B,
\end{align*}
where $\mathbf{\hat f}_A = -Q(m_{A2}\mathbf{e}_2 + m_{A3}\mathbf{e}_3) - \delta\mathbf{m}_B + \mathbf{h}_e$ and $\mathbf{\hat f}_B = -Q(m_{B2}\mathbf{e}_2 + m_{B3}\mathbf{e}_3)-\delta\mathbf{m}_A + \mathbf{h}_e$. The GSPM solves \eqref{equ:LLG2_dimensionless} in three steps:
\begin{itemize}
\item Implicit Gauss-Seidel:
\begin{align*}
{g_A}_i^n & = (I - \Delta t\epsilon\Delta_h)^{-1}({m_A}_i^n + \Delta t{\hat f}_{Ai}^n),\ \ i = 2,3,\\
{g_A}_i^* & = (I - \Delta t\epsilon\Delta_h)^{-1}({m_A}_i^* + \Delta t{\hat f}_{Ai}^*),\ \ i = 1,2,\\
\begin{pmatrix}
{m_A}_1^* \\{m_A}_2^* \\{m_A}_3^*
\end{pmatrix} & =
\begin{pmatrix}
{m_A}_1^n + ({g_A}_2^n{m_A}_3^n - {g_A}_3^n{m_A}_2^n) \\
{m_A}_2^n + ({g_A}_3^n{m_A}_1^* - {g_A}_1^*{m_A}_3^n) \\
{m_A}_3^n + ({g_A}_1^*{m_A}_2^* - {g_A}_2^*{m_A}_1^*)
\end{pmatrix}, \\
  \mathbf{\hat f}_B^* & = -Q(m_{B2}\mathbf{e}_2 + m_{B3}\mathbf{e}_3)-\delta\mathbf{m}_A^* + \mathbf{h}_e, \\
{g_B}_i^n & = (I - \Delta t\epsilon\Delta_h)^{-1}({m_B}_i^n + \Delta t{\hat f}_{Bi}^*),\ \ i = 2,3, \\
{g_B}_i^* & = (I - \Delta t\epsilon\Delta_h)^{-1}({m_B}_i^* + \Delta t{\hat f}_{Bi}^*),\ \ i = 1,2, \\
\begin{pmatrix}
{m_B}_1^* \\{m_B}_2^* \\{m_B}_3^*
\end{pmatrix} & =
\begin{pmatrix}
{m_B}_1^n + ({g_B}_2^n{m_B}_3^n - {g_B}_3^n{m_B}_2^n) \\
{m_B}_2^n + ({g_B}_3^n{m_B}_1^* - {g_B}_1^*{m_B}_3^n) \\
{m_B}_3^n + ({g_B}_1^*{m_B}_2^* - {g_B}_2^*{m_B}_1^*)
\end{pmatrix},
\end{align*}
\begin{align*}
  \mathbf{\hat f}_A^* &= -Q(m_{A2}^*\mathbf{e}_2 + m_{A3}^*\mathbf{e}_3)-\delta\mathbf{m}_B^* + \mathbf{h}_e, \nonumber\\
  \mathbf{\hat f}_B^{**} &= -Q(m_{B2}^*\mathbf{e}_2 + m_{B3}^*\mathbf{e}_3)-\delta\mathbf{m}_A^{*} + \mathbf{h}_e.
\end{align*}
\item Heat flow without constraints:
\begin{align*}
\begin{pmatrix}
{m_A}_1^{**} \\{m_A}_2^{**} \\{m_A}_3^{**}
\end{pmatrix} &=
\begin{pmatrix}
{m_A}_1^{*} + \alpha\Delta t(\epsilon\Delta_h{m_A}_1^{**} + {\hat f}_{A1}^*) \\
{m_A}_2^{*} + \alpha\Delta t(\epsilon\Delta_h{m_A}_2^{**} + {\hat f}_{A2}^*) \\
{m_A}_3^{*} + \alpha\Delta t(\epsilon\Delta_h{m_A}_3^{**} + {\hat f}_{A3}^*)
\end{pmatrix},\\
\begin{pmatrix}
{m_B}_1^{**} \\{m_B}_2^{**} \\{m_B}_3^{**}
\end{pmatrix} &=
\begin{pmatrix}
{m_B}_1^{*} + \alpha\Delta ts^2(\epsilon\Delta_h{m_B}_1^{**} + {\hat f}_{B1}^{**}) \\
{m_B}_2^{*} + \alpha\Delta ts^2(\epsilon\Delta_h{m_B}_2^{**} + {\hat f}_{B2}^{**}) \\
{m_B}_3^{*} + \alpha\Delta ts^2(\epsilon\Delta_h{m_B}_3^{**} + {\hat f}_{B3}^{**})
\end{pmatrix}.
\end{align*}
\item Projection onto $S^2$:
\begin{align*}
\begin{pmatrix}
{m_A}_1^{n+1} \\{m_A}_2^{n+1} \\{m_A}_3^{n+1}
\end{pmatrix} = \frac 1{|\mathbf{m}_A^{**}|}
\begin{pmatrix}
{m_A}_1^{**} \\{m_A}_2^{**} \\{m_A}_3^{**}
\end{pmatrix},\\
\begin{pmatrix}
{m_B}_1^{n+1} \\{m_B}_2^{n+1} \\{m_B}_3^{n+1}
\end{pmatrix} = \frac s{|\mathbf{m}_B^{**}|}
\begin{pmatrix}
{m_B}_1^{**} \\{m_B}_2^{**} \\{m_B}_3^{**}
\end{pmatrix}.
\end{align*}
\end{itemize}
Note that $0<s\leqslant1$ within the definition $|\mathbf{m}_B| = s$.
$s=1$ is for AFMs and $s<1$ is for ferrimagnetics.

\subsection{Scheme A}\label{subsec:schemeA}
Both Scheme A and Scheme B are based on improved GSPMs for ferromagnetics \cite{panchi2019GSPM}.
In Scheme A, we do not treat the gyromagnetic term and the damping term separately. Instead, the implicit Gauss-Seidel method is applied to the gyromagnetic term and the damping term simultaneously, and a projection step follows up.
\begin{itemize}
\item Implicit Gauss-Seidel step:
\begin{align*}
g_{Ai}^n &= (I-\Delta t\epsilon\Delta)^{-1}(m_{Ai}^n+\Delta t\hat{f}_{Ai}^n),\ \ i = 1,2,3,\\
g_{Ai}^* &= (I-\Delta t\epsilon\Delta)^{-1}(m_{Ai}^*+\Delta t\hat{f}_{Ai}^*),\ \ i = 1,2,\\
\end{align*}
\begin{align*}
m_{A1}^* &= m_{A1}^n -(m_{A2}^ng_{A3}^n - m_{A3}^ng_{A2}^n)-\alpha(m_{A1}^ng_{A1}^n + m_{A2}^ng_{A2}^n + m_{A3}^ng_{A3}^n)m_{A1}^n + \alpha g_{A1}^n,\\
m_{A2}^* &= m_{A2}^n -(m_{A3}^ng_{A1}^* - m_{A1}^*g_{A3}^n)-\alpha(m_{A1}^*g_{A1}^* + m_{A2}^ng_{A2}^n + m_{A3}^ng_{A3}^n)m_{A2}^n + \alpha g_{A2}^n,\\
m_{A3}^* &= m_{A3}^n -(m_{A1}^*g_{A2}^* - m_{A2}^*g_{A1}^*)-\alpha(m_{A1}^*g_{A1}^* + m_{A2}^*g_{A2}^* +
m_{A3}^ng_{A3}^n)m_{A3}^n + \alpha g_{A3}^n,
\end{align*}
\begin{align*}
  g_{Bi}^n &= (I-\Delta t\epsilon\Delta)^{-1}(m_{Bi}^n+\Delta t\hat{f}_{Bi}^*),\ \ i = 1,2,3,\\
  g_{Bi}^* &= (I-\Delta t\epsilon\Delta)^{-1}(m_{Bi}^*+\Delta t\hat{f}_{Bi}^*),\ \ i = 1,2,
\end{align*}
\begin{align*}
m_{B1}^* &= m_{B1}^n -(m_{B2}^ng_{B3}^n - m_{B3}^ng_{B2}^n)-\alpha(m_{B1}^ng_{B1}^n + m_{B2}^ng_{B2}^n + m_{B3}^ng_{B3}^n)m_{B1}^n + \alpha s^2g_{B1}^n,\\
m_{B2}^* &= m_{B2}^n -(m_{B3}^ng_{B1}^* - m_{B1}^*g_{B3}^n)-\alpha(m_{B1}^*g_{B1}^* + m_{B2}^ng_{B2}^n + m_{B3}^ng_{B3}^n)m_{B2}^n + \alpha s^2g_{B2}^n,\\
m_{B3}^* &= m_{B3}^n -(m_{B1}^*g_{B2}^* - m_{B2}^*g_{B1}^*)-\alpha(m_{B1}^*g_{B1}^* + m_{B2}^*g_{B2}^* + m_{B3}^ng_{B3}^n)m_{B3}^n + \alpha s^2g_{B3}^n.
\end{align*}
\item Projection step:
\begin{align*}
\begin{pmatrix}
{m_A}_1^{n+1} \\{m_A}_2^{n+1} \\{m_A}_3^{n+1}
\end{pmatrix} = \frac 1{|\mathbf{m}_A^{*}|}
\begin{pmatrix}
{m_A}_1^{*} \\{m_A}_2^{*} \\{m_A}_3^{*}
\end{pmatrix},\\
\begin{pmatrix}
{m_B}_1^{n+1} \\{m_B}_2^{n+1} \\{m_B}_3^{n+1}
\end{pmatrix} = \frac s{|\mathbf{m}_B^{*}|}
\begin{pmatrix}
{m_B}_1^{*} \\{m_B}_2^{*} \\{m_B}_3^{*}
\end{pmatrix}.
\end{align*}
\end{itemize}
\subsection{Scheme B}\label{subsec:schemeB}

Scheme B reduces the computational cost further by the introduction of two sets of approximations.
At each step, one set of solution is updated in the implicit Gauss-Seidel step and
the other is updated in the projection step.
\begin{itemize}
\item Implicit Gauss-Seidel step:
\begin{align*}
g_{Ai}^{n+1} &= (I-\Delta t\epsilon\Delta_h)^{-1}(m_{Ai}^*+\Delta t\hat{f}_{Ai}^*),\ \ i = 1,2,3,
\end{align*}
\begin{align*}
m_{A1}^* = &m_{A1}^n -(m_{A2}^ng_{A3}^n - m_{A3}^ng_{A2}^n)-\alpha(m_{A1}^ng_{A1}^n + m_{A2}^ng_{A2}^n + m_{A3}^ng_{A3}^n)m_{A1}^n + \\
&\alpha((m_{A1}^n)^2 + (m_{A2}^n)^2 + (m_{A3}^n)^2)g_{A1}^n,\\
m_{A2}^* = &m_{A2}^n -(m_{A3}^ng_{A1}^{*} - m_{A1}^{*}g_{A3}^n)-\alpha(m_{A1}^{*}g_{A1}^{n+1} + m_{A2}^ng_{A2}^n + m_{A3}^ng_{A3}^n)m_{A2}^n + \\
&\alpha((m_{A1}^{*})^2 + (m_{A2}^n)^2 + (m_{A3}^n)^2)g_{A2}^n,\\
m_{A3}^* = &m_{A3}^n -(m_{A1}^{*}g_{A2}^{n+1} - m_{A2}^{*}g_{A1}^{n+1})-\alpha(m_{A1}^{*}g_{A1}^{n+1} + m_{A2}^{*}g_{A2}^{n+1} + m_{A3}^ng_{A3}^n)m_{A3}^n + \\
&\alpha((m_{A1}^{*})^2 + (m_{A2}^{*})^2 + (m_{A3}^n)^2)g_{A3}^n,
\end{align*}
\begin{align*}
  g_{Bi}^{n+1} &= (I-\Delta t\epsilon\Delta_h)^{-1}(m_{Bi}^*+\Delta t\hat{f}_{Bi}^*),\ \ i = 1,2,3,
\end{align*}
\begin{align*}
m_{B1}^* = &m_{B1}^n -(m_{B2}^ng_{B3}^n - m_{B3}^ng_{B2}^n)-\alpha(m_{B1}^ng_{B1}^n + m_{B2}^ng_{B2}^n + m_{B3}^ng_{B3}^n)m_{B1}^n + \\
&\alpha((m_{B1}^n)^2 + (m_{B2}^n)^2 + (m_{B3}^n)^2)g_{B1}^n,\\
m_{B2}^*= &m_{B2}^n -(m_{B3}^ng_{B1}^{n+1} - m_{B1}^{*}g_{B3}^n)-\alpha(m_{B1}^{*}g_{B1}^{n+1} + m_{B2}^ng_{B2}^n + m_{B3}^ng_{B3}^n)m_{B2}^n + \\
&\alpha((m_{B1}^{*})^2 + (m_{B2}^n)^2 + (m_{B3}^n)^2)g_{B2}^n,\\
m_{B3}^*= &m_{B3}^n -(m_{B1}^{*}g_{B2}^{n+1} - m_{B2}^{*}g_{B1}^{n+1})-\alpha(m_{B1}^{*}g_{B1}^{n+1} + m_{B2}^{*}g_{B2}^{n+1} + m_{B3}^ng_{B3}^n)m_{B3}^n + \\
&\alpha((m_{B1}^{*})^2 + (m_{B2}^{*})^2 + (m_{B3}^n)^2)g_{B3}^n.
\end{align*}

\item Projection step:
\begin{align*}
\begin{pmatrix}
{m_A}_1^{n+1} \\{m_A}_2^{n+1} \\{m_A}_3^{n+1}
\end{pmatrix} = \frac 1{|\mathbf{m}_A^{*}|}
\begin{pmatrix}
{m_A}_1^{*} \\{m_A}_2^{*} \\{m_A}_3^{*}
\end{pmatrix},\\
\begin{pmatrix}
{m_B}_1^{n+1} \\{m_B}_2^{n+1} \\{m_B}_3^{n+1}
\end{pmatrix} = \frac s{|\mathbf{m}_B^{*}|}
\begin{pmatrix}
{m_B}_1^{*} \\{m_B}_2^{*} \\{m_B}_3^{*}
\end{pmatrix}.
\end{align*}
\end{itemize}

The above three GSPMs for AFMs and ferrimagntics have different computational complexity originating from the number of linear systems of
equations with constant coefficients to be solved at each step. For comparison, we list the number of linear systems to be solved at each step in \Cref{tab:effective}. It is easy to see that the computational complexity of GSPM for AFMs is only doubled compared to that for FMs.
\begin{table}[htbp]
  \centering
  \caption{Number of linear systems of equations to be solved at each step for three GSPMs.}
	\begin{tabular}{c|c|c}
		\hline
		Scheme & Number of linear systems for AFMs & Number of linear systems for FMs \\
		\hline
		GSPM & $14$ & $7$\\
		Scheme A & $10$ & $5$\\
		Scheme B & $6$ & $3$ \\
		\hline
	\end{tabular}
    \label{tab:effective}
\end{table}
\section{Accuracy check}\label{sec:experiments}

In this section, by a series of examples in both 1D and 3D, we show the accuracy of GSPMs
for AFMs and ferrimagnets. For convenience, the model used here is
\begin{equation}
\left\{
\begin{aligned}
\frac{\partial\mathbf{m}_A}{\partial t} &= -\mathbf{m}_A\times\mathbf{h}_A
-\alpha\mathbf{m}_A\times(\mathbf{m}_A\times\mathbf{h}_A) + \mathbf{f}_A\\
\frac{\partial\mathbf{m}_B}{\partial t} &= -\mathbf{m}_B\times\mathbf{h}_B
-\alpha\mathbf{m}_B\times(\mathbf{m}_B\times\mathbf{h}_B) + \mathbf{f}_B
\end{aligned}
\right.  \label{equ:for_accuracy_test}
\end{equation}
with $\mathbf{h}_A = \Delta\mathbf{m}_A + \delta\mathbf{m}_B$, $\mathbf{h}_B = \Delta\mathbf{m}_B + \delta\mathbf{m}_A$,
and $\mathbf{f}_{\lambda}$ are forcing terms specified by exact solutions.

\begin{example}[1D]\label{eg:1daccuracy}
  Consider a set of orthogonal solutions for \eqref{equ:for_accuracy_test} in $\Omega=[0,1]$:
  \begin{align*}
    \mathbf{m}_A &= (\cos(x^2(1-x)^2)\sin(t), \sin(x^2(1-x)^2)\sin(t), \cos(t)), \nonumber\\
    \mathbf{m}_B &= (s\cos(x^2(1-x)^2)\cos(t), s\sin(x^2(1-x)^2)\cos(t), -s\sin(t)).
  \end{align*}
 Parameters are $\alpha = 0.1$, $\delta = 2.0$ and $T=1.0e-03$. The error is defined $\|\mathbf{m}_e - \mathbf{m}_h\|_{\infty}$
 with $\mathbf{m}_h$ being the numerical solutions and $\mathbf{m}_e$ being the exact solution.
 The accuracy of GSPMs is $O(\Delta t + \Delta x^2)$ as shown in \Cref{table:1DTemperoaltable,table:1DSpatialtable}.
  \begin{table}[htbp]
    \centering
    \caption{Accuracy with respect to the temporal step size in 1D ($\Delta x = 0.001$).}
    \begin{tabular}{c|c|c|c|c|c|c}
      \hline
      $\Delta t$ & \tabincell{c}{GSPM\\(s = 1.0)} & \tabincell{c}{GSPM\\(s = 0.8)} & \tabincell{c}{Scheme A\\(s = 1.0)} & \tabincell{c}{Scheme A\\(s = 0.8)} & \tabincell{c}{Scheme B\\(s = 1.0)} & \tabincell{c}{Scheme B\\(s = 0.8)} \\
      \hline
      T/1000 & 2.7388e-07 & 2.0377e-07 & 3.2269e-07 & 2.2854e-07 & 3.2270e-07 & 2.2854e-07 \\
      T/500 & 5.3937e-07 & 4.0141e-07 & 6.3259e-07 & 4.4798e-07 & 6.3260e-07 & 4.4798e-07\\
      T/250 & 1.0740e-06 & 7.9971e-07 & 1.2525e-06 & 8.8680e-07 & 1.2525e-06 & 8.8680e-07 \\
      T/125 & 2.1541e-06 & 1.6055e-06 & 2.4925e-06 & 1.7643e-06 & 2.4925e-06 & 1.7642e-06\\
      \hline
      order & 0.99 & 0.99 & 0.98 & 0.98 & 0.98 & 0.98\\
      \hline
    \end{tabular}
    \label{table:1DTemperoaltable}
  \end{table}

  \begin{table}[htbp]
    \centering
    \caption{Accuracy with respect to the spatial step size in 1D ($\Delta t = 1.0e-09$).}
    \begin{tabular}{c|c|c|c|c|c|c}
      \hline
      $\Delta x$ & \tabincell{c}{GSPM\\(s = 1.0)} & \tabincell{c}{GSPM\\(s = 0.8)} & \tabincell{c}{Scheme A\\(s = 1.0)} & \tabincell{c}{Scheme A\\(s = 0.8)} & \tabincell{c}{Scheme B\\(s = 1.0)} & \tabincell{c}{Scheme B\\(s = 0.8)} \\
      \hline
      0.001 & 1.3063e-08 & 9.2823e-09 & 1.3110e-08 & 9.3084e-09 & 1.3121e-08 & 9.3120e-09 \\
      0.002 & 5.0500e-08 & 3.5764e-08 & 5.0546e-08 & 3.5790e-08 & 5.0557e-08 & 3.5794e-08 \\
      0.004 & 1.9341e-07 & 1.3623e-07 & 1.9345e-07 & 1.3626e-07 & 1.9346e-07 & 1.3626e-07 \\
      0.008 & 7.1142e-07 & 4.9543e-07 & 7.1146e-07 & 4.9545e-07 & 7.1147e-07 & 4.9545e-07 \\
      \hline
      order & 1.92 & 1.92 & 1.92 & 1.91 & 1.91 & 1.91\\
      \hline
    \end{tabular}
    \label{table:1DSpatialtable}
  \end{table}
\end{example}

\begin{example}[3D]\label{eg:3daccuracy}
  As in the 1D case, we use a set of orthogonal solutions for \eqref{equ:for_accuracy_test} in $\Omega=[0,1]\times[0,1]\times[0,1]$ for the spatial accuracy and $\Omega=[0,0.2]\times[0,0.1]\times[0, 0.02]$ with a $64\times32\times5$ mesh for the temporal accuracy
  \begin{align*}
    \mathbf{m}_A &= (\cos({\bar x\bar y\bar z})\sin(t), \sin({\bar x\bar y\bar z})\sin(t), \cos(t)), \nonumber\\
    \mathbf{m}_B &= (s\cos({\bar x\bar y\bar z})\cos(t), s\sin({\bar x\bar y\bar z})\cos(t), -s\sin(t))
  \end{align*}
  with ${\bar x} = x^2(1-x)^2$, ${\bar y} = y^2(1-y)^2$, ${\bar z} = z^2(1-z)^2$. Parameters are $\alpha = 0.1$ and $\delta = 2.0$.
  The first order accuracy in time and the second order accuracy in space are shown in \Cref{table:3DTemperoaltable,table:3DSpatialtable}.

  \begin{table}[htbp]
    \centering
    \caption{Accuracy with respect to the temporal step size in 3D ($\Delta x = 0.001$ and $T = 1.0e-06$).}
    \begin{tabular}{c|c|c|c|c|c|c}
      \hline
      $\Delta t$ & \tabincell{c}{GSPM\\(s = 1.0)} & \tabincell{c}{GSPM\\(s = 0.8)} & \tabincell{c}{Scheme A\\(s = 1.0)} & \tabincell{c}{Scheme A\\(s = 0.8)} & \tabincell{c}{Scheme B\\(s = 1.0)} & \tabincell{c}{Scheme B\\(s = 0.8)} \\
      \hline
      T/25 & 4.0001e-08 & 4.0000e-08 & 4.0001e-08 & 4.0000e-08 & 4.0001e-08 & 4.0000e-08 \\
      T/50 & 2.0001e-08 & 2.0000e-08 & 2.0001e-08 & 2.0000e-08 & 2.0001e-08 & 2.0000e-08\\
      T/100 & 1.0001e-08 & 1.0000e-08 & 1.0001e-08 & 1.0000e-08 & 1.0001e-08 & 1.0000e-08 \\
      T/200 & 5.0012e-09 & 5.0000e-09 & 5.0012e-09 & 5.0000e-09 & 5.0012e-09 & 5.0000e-09\\
      \hline
      order & 1.00 & 1.00 & 1.00 & 1.00 & 1.00 & 1.00\\
      \hline
    \end{tabular}
    \label{table:3DTemperoaltable}
  \end{table}

  \begin{table}[htbp]
    \centering
    \caption{Accuracy with respect to spatial step size in 3D ($\Delta t = 1.0e-09$ and $T = 1.0e-04$).}
    \begin{tabular}{c|c|c|c|c|c|c}
      \hline
      $\Delta x$ & \tabincell{c}{GSPM\\(s = 1.0)} & \tabincell{c}{GSPM\\(s = 0.8)} & \tabincell{c}{Scheme A\\(s = 1.0)} & \tabincell{c}{Scheme A\\(s = 0.8)} & \tabincell{c}{Scheme B\\(s = 1.0)} & \tabincell{c}{Scheme B\\(s = 0.8)} \\
      \hline
      1/6 & 5.9150e-08 & 3.8020e-08 & 5.9150e-08 & 3.8020e-08 & 5.9150e-08 & 3.8020e-08 \\
      1/8 & 3.5375e-08 & 2.2803e-08 & 3.5375e-08 & 2.2803e-08 & 3.5375e-08 & 2.2803e-08\\
      1/10 & 2.3507e-08 & 1.5206e-08 & 2.3507e-08 & 1.5206e-08 & 2.3507e-08 & 1.5206e-08 \\
      1/12 & 1.6824e-08 & 1.0928e-08 & 1.6824e-08 & 1.0928e-08 & 1.6824e-08 & 1.0928e-08\\
      \hline
      order & 1.81 & 1.80 & 1.81 & 1.80 & 1.81 & 1.80\\
      \hline
    \end{tabular}
    \label{table:3DSpatialtable}
  \end{table}
\end{example}

\begin{example}[Accuracy with respect to $\delta$]\label{eg:3dAFMparameter}
Using the same setup as in \Cref{eg:3daccuracy}, we also show that GSPMs work well for the antiferromagnetic exchange parameter $\delta$
over a wide range of values in \cref{Unconditionalstablewith_delta}.

 \begin{figure}[htbp]
    \centering
    \subfloat[Time]{\label{temporalwithdelta}\includegraphics[width=2.5in]{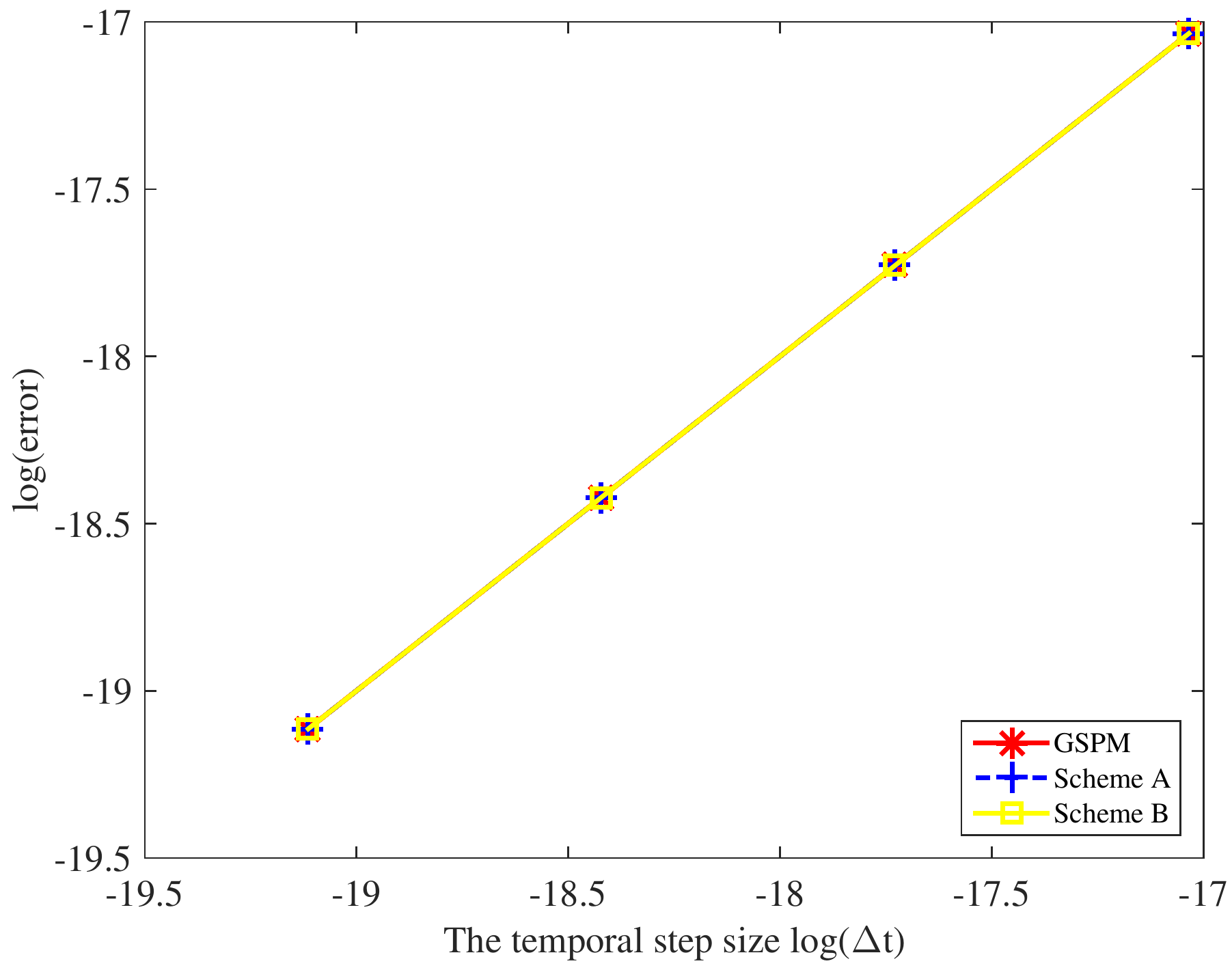}}
    \subfloat[Space]{\label{spatialwithdelta}\includegraphics[width=2.5in]{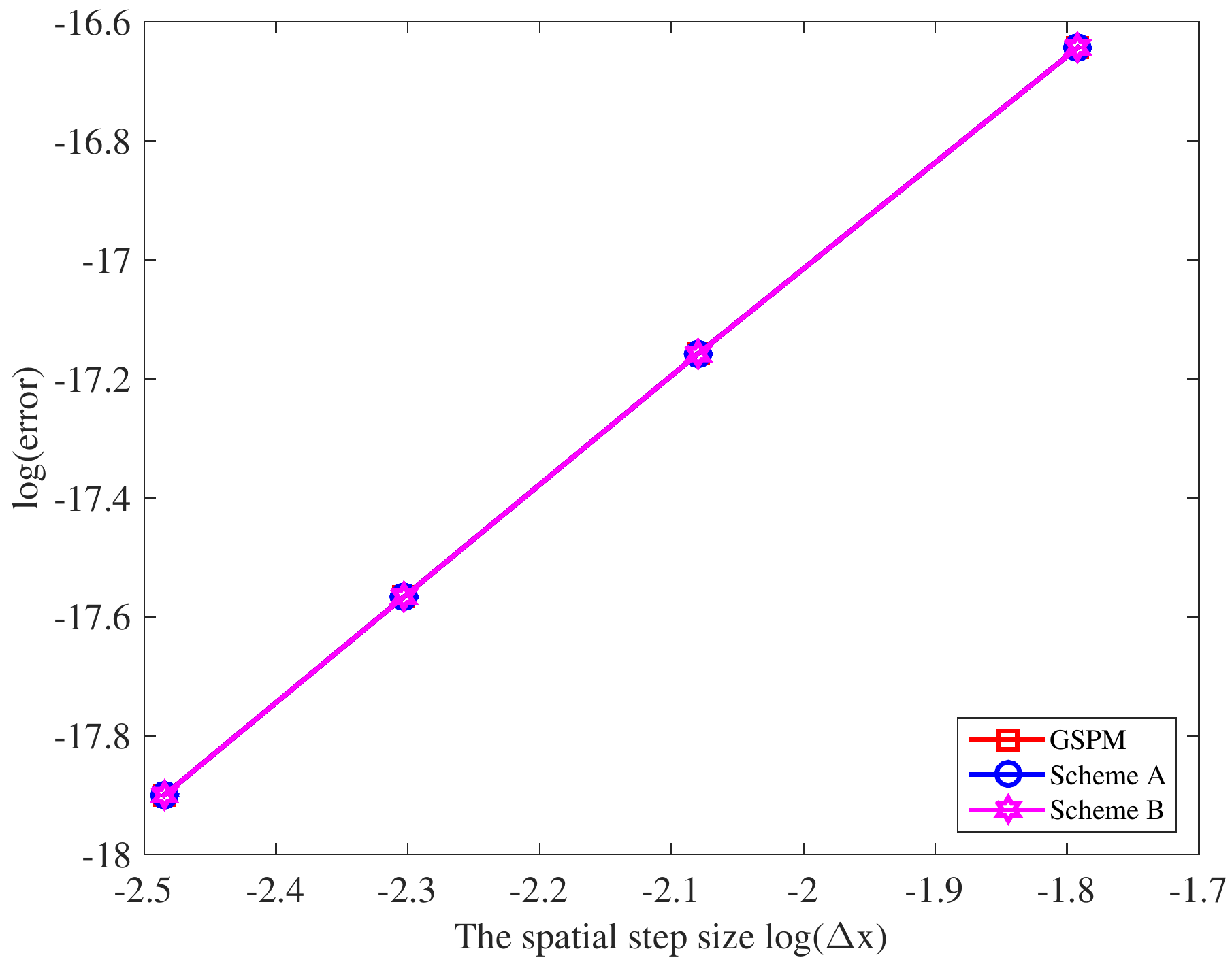}}
    \caption{Performance of three GSPMs with respect to the antiferromagnetic exchange parameter $\delta$ in both time and space.
    The antiferromagnetic exchange parameter $\delta = \pm0.1, \pm5.0, \pm100.0$.}
    \label{Unconditionalstablewith_delta}
  \end{figure}
\end{example}

\section{Femtosecond magnetization dynamics in antiferromagnets}\label{sec:simulaions}

The presence of the antiferromagnetic exchange term makes magnetization dynamics richer compared with ferromagnetic counterparts.
In the simulations, the size of the thin film material is $100\;\mathrm{nm}\times100\;\mathrm{nm}\times10\;\mathrm{nm}$ and the
grid size is $2\;\mathrm{nm}\times2\;\mathrm{nm}\times2\;\mathrm{nm}$.
Physical parameters \cite{Puliafito2019Model} are listed in \cref{table:physicalparameter}. For the given parameters, we have $\delta/q = 240$ and
$\delta/\epsilon = 4.8\times10^4$. Therefore for micromagnetics simulation in AFMs, the temporal stepsize is limited by
the antiferromagnetic exchange parameter.

\begin{table}[htbp]
\centering
\caption{Physical parameters used in the simulation.}
  \begin{tabular}{|c|c|c|}
    \hline
    parameters & value & unit \\
    \hline
    \hline
    $\alpha$ & $0.05$ & $-$\\
    \hline
    $a$ & $0.5\times10^{-9}$ & m \\
    \hline
    $M_s$ & $4.0\times10^5$ & $\mathrm{A}/\mathrm{m}$ \\
    \hline
    $K_U$ & $1.0\times10^5$ & $\mathrm{J}/\mathrm{m}^3$\\
    \hline
    $A$ & $5.0\times10^{-12}$ & $\mathrm{J}/\mathrm{m}$\\
    \hline
    $A_{AFM}$ & $\pm3.0\times10^{-12}$ & $\mathrm{J}/\mathrm{m}$ \\
    \hline
    $\gamma$ & $1.76\times10^{11}$ & $(\mathrm{Ts})^{-1}$\\
    \hline
  \end{tabular}
  \label{table:physicalparameter}
\end{table}

\subsection{Antiferromagnetic/ferromagnetic stable phase}\label{subsec:phase}

In this case, we choose the initial state as $\mathbf{m}_A = (1, 0, 0)^T$, and $\mathbf{m}_B = (0, 1, 0)^T$,
and solve \eqref{equ:LLG2_dimensionless}.
Note that the easy-axis direction is the $x\textrm{-}axis$. In the absence of an external field, the system
shall relax to a ferromagnetic phase or an antiferromagnetic phase, depending on the sign of the antiferromagnetic
exchange parameter $\delta$. Since the temporal stepsize is limited by the antiferromagnetic exchange parameter,
we observe that all GSPMs lose their accuracy if the stepsize is much larger than the femtosecond scale. More
seriously, the numerical solution of Scheme B blows up with stepsize $\Delta t = 1\;\textrm{ps}$.
Therefore, in what follows, we set $\Delta t = 1\;\textrm{fs}$ to get reliable numerical solutions.

There are two characteristic time scales of the magnetization dynamics in AFMs when the external filed
is applied. The first one is due to the antiferromagnetic exchange,
which yields an intermediate antiferromagetic/ferromagnetic state at $\sim 1\;\textrm{ps}$. The second one is due to the
magnetic anisotropy, which yields a stable antiferromagnetic/ferromagnetic states at $\sim 10\;\textrm{ps}$ or longer.
\cref{AFMstablewith_delta} visualizes an intermediate antiferromagnetic state at $2\;\textrm{ps}$
and an stable antiferromagnetic state at $6\;\textrm{ps}$ when the antiferromagnetic exchange parameter $A_{AFM} = 3.0\times10^{-12}\;\textrm{J}/\textrm{m}$.
\cref{FMstablewith_delta} visualizes an intermediate ferromagnetic state at $2\;\textrm{ps}$
and a stable antiferromagnetic state at $500\;\textrm{ps}$ when $A_{AFM} = -3.0\times10^{-12}\;\textrm{J}/\textrm{m}$.

\begin{figure}[htbp]
    \centering
    \subfloat[Intermediate state]{\includegraphics[width=2.5in]{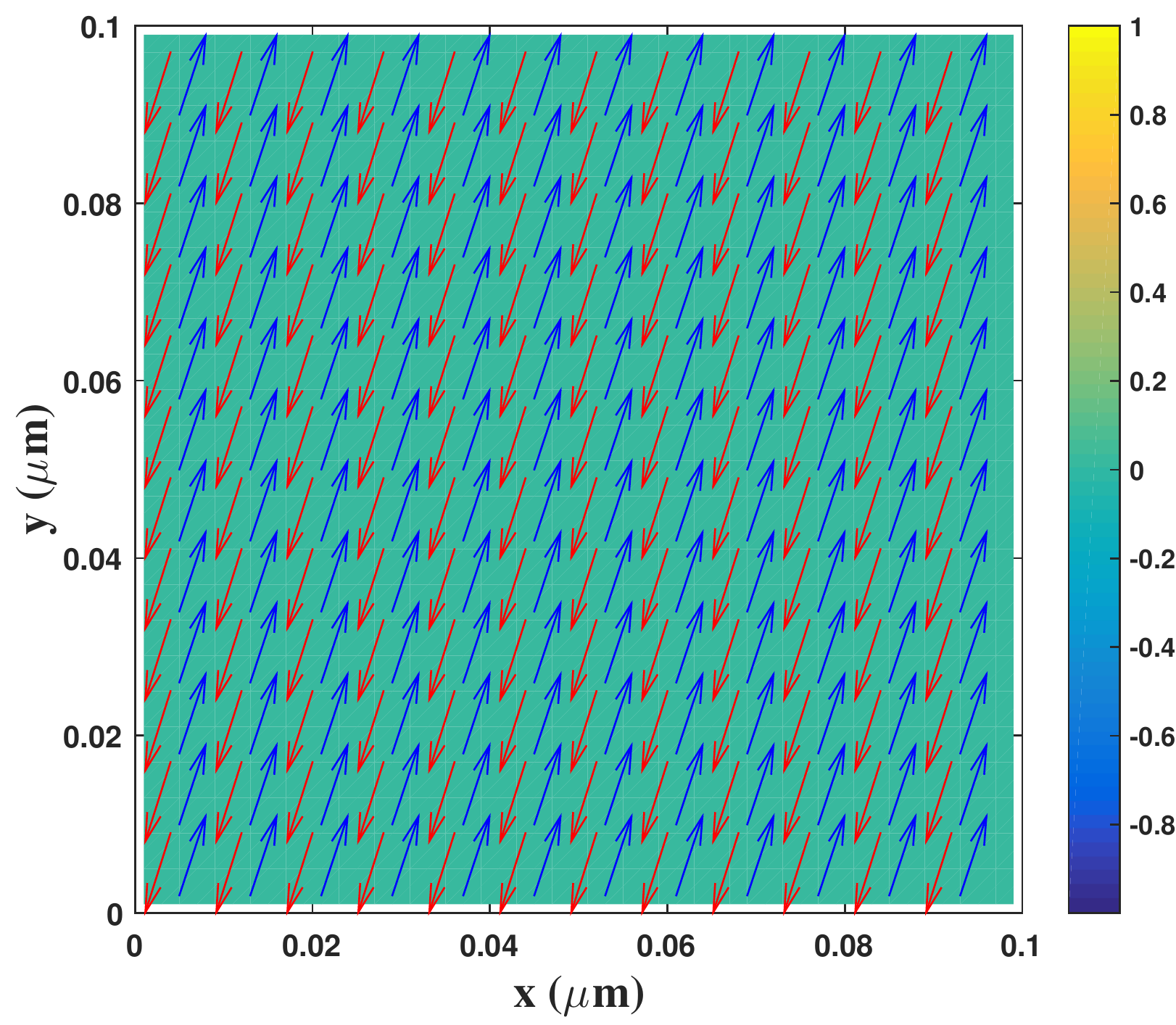}}
    \subfloat[Antiferromagnetic stable state]{\includegraphics[width=2.5in]{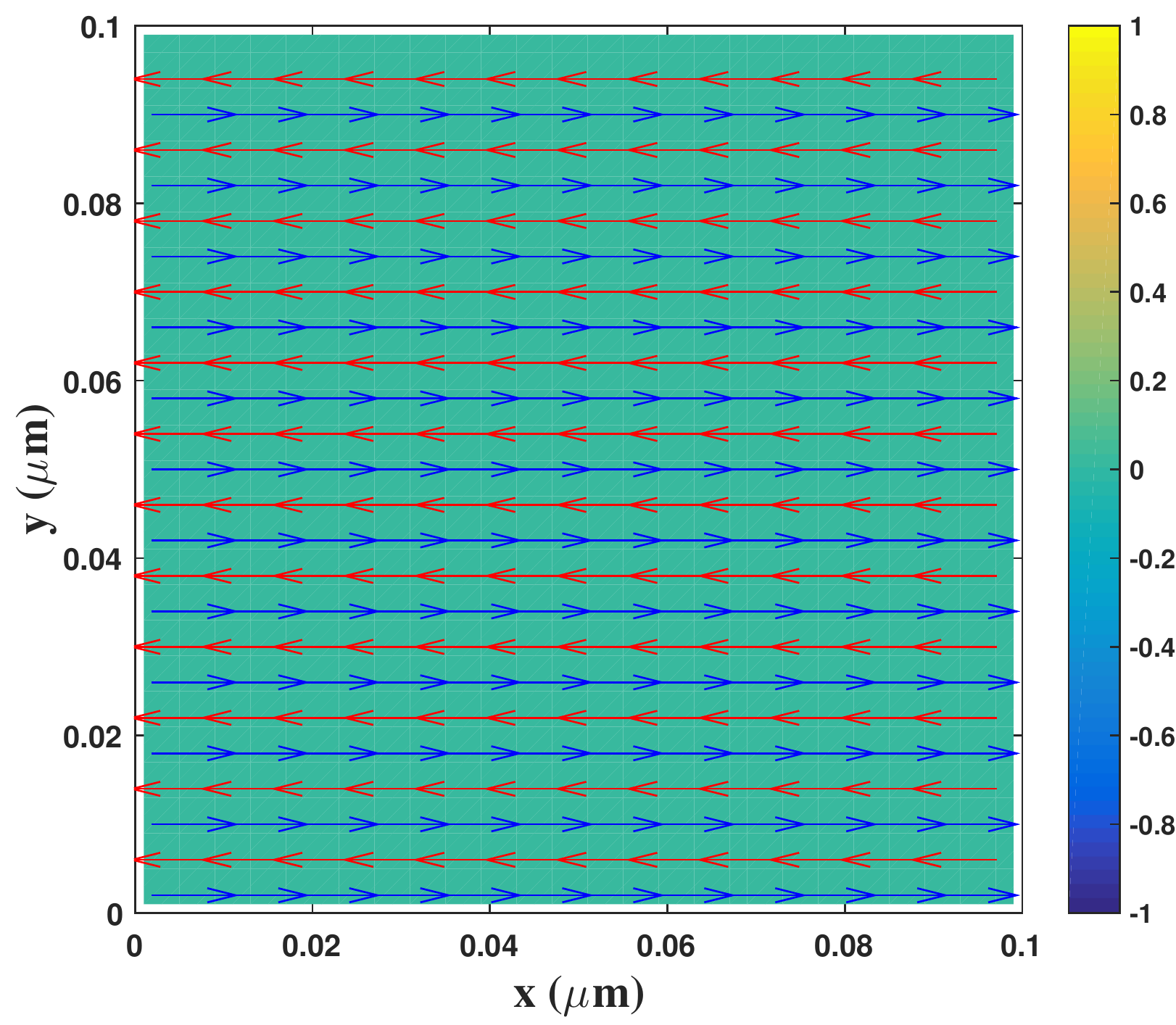}}
    \caption{Arrow plot of the magnetization in the centered slice along the $z$ direction. Arrows are plotted by the first two components of the magnetization and the color is plotted by the third component $(m_{A3}+m_{B3})/2$. We set $\Delta t = 1\;\textrm{fs}$ and $T = 40\;\textrm{ps}$ although it stabilizes at $6\;\textrm{ps}$. (a) An intermediate antiferromagnetic state; (b) Stable antiferromagnetic state.}
    \label{AFMstablewith_delta}
\end{figure}
\begin{figure}[htbp]
    \centering
    \subfloat[Intermediate state]{\includegraphics[width=2.5in]{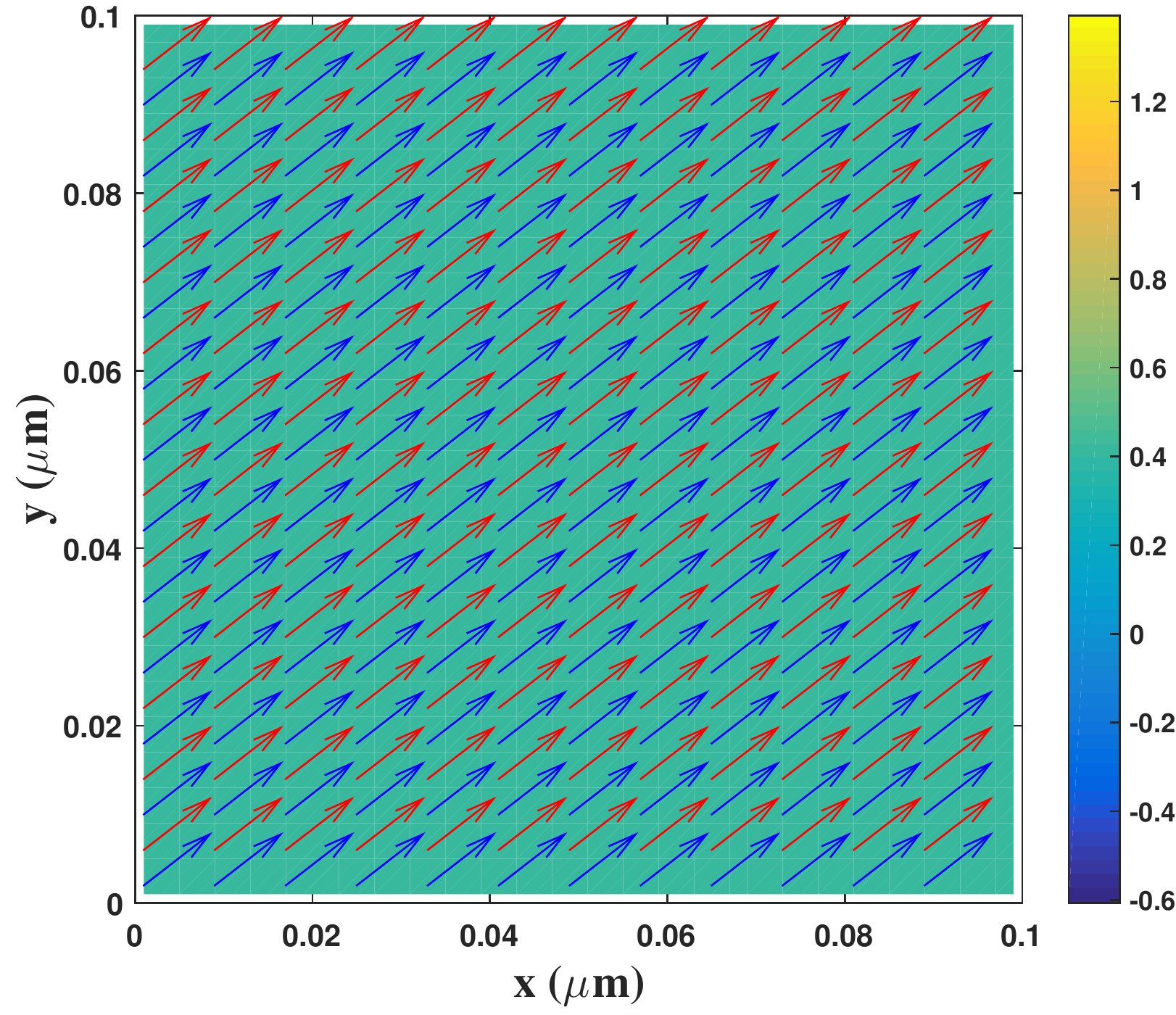}}
    \subfloat[Ferromagnetic stable state]{\label{FMstable}\includegraphics[width=2.5in]{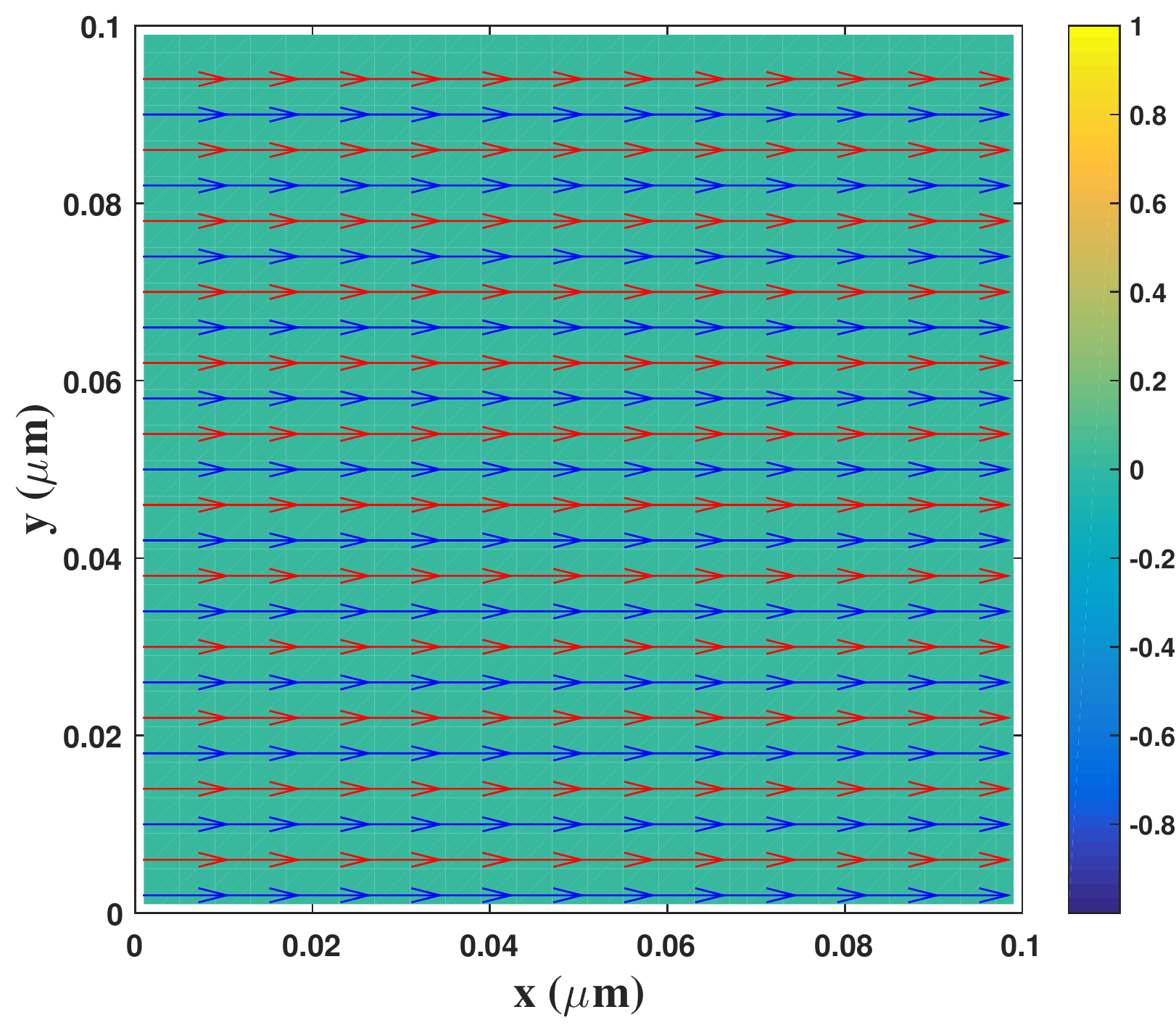}}
    \caption{Arrow plot of the magnetization in the centered slice along the $z$ direction. Arrows are plotted by the first two components of the magnetization	and the color is plotted by the third component $(m_{A3}+m_{B3})/2$. We set $\Delta t = 1\;\textrm{fs}$ and $T = 2\;\textrm{ns}$ although it stabilizes at $500\;\textrm{ps}$. (a) An intermediate antiferromagnetic state; (b) Stable antiferromagnetic state.}
    \label{FMstablewith_delta}
\end{figure}

In addition, we record the system energy with respect to time for $A_{AFM} = 3.0\times10^{-12}\;\textrm{J}/\textrm{m}$ (antiferromagnetism),
$A_{AFM} = -3.0\times10^{-12}\;\textrm{J}/\textrm{m}$ (ferromagnetism), and $A_{AFM} = 0$ (no coupling) in \Cref{energy} with random initial conditions for $\m_{A}$ and $\m_{B}$. As mentioned above, two distinct time scales are observed for nonzero antiferromagnetic exchange parameter
while the shorter time scale is missing if $A_{AFM} = 0$.
Although the energy in \Cref{energy} seems to be saturated, the stable ferromagnetic state has not been obtained yet. It actually takes longer to achieve
the stable state as shown in \Cref{FMstable}.

\begin{figure}[htbp]
  \centering
  \includegraphics[width=6in]{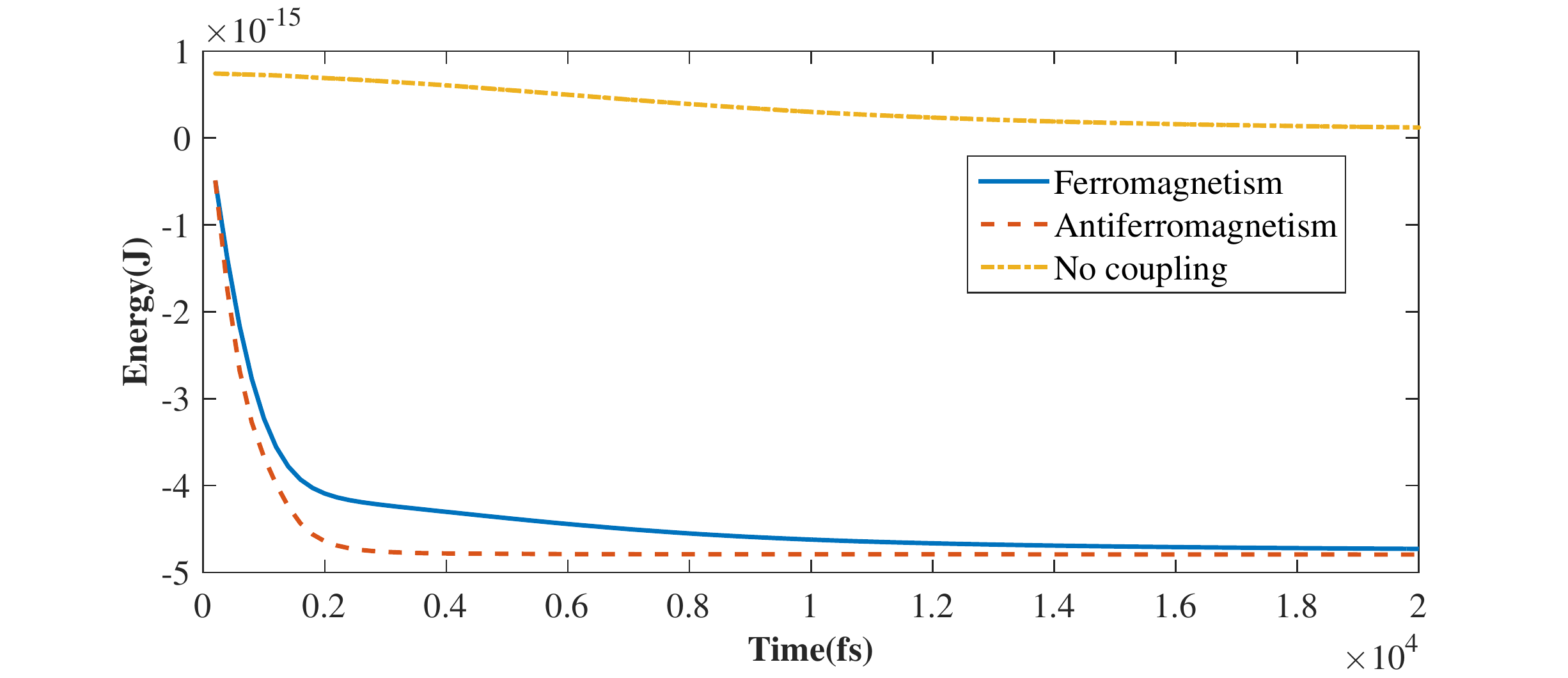}
  \caption{The system energy as a function of time for random initial conditions when $A_{AFM} = 3\times10^{-12}$ (antiferromagnetism),
  	$A_{AFM} = -3\times10^{-12}$ (ferromagnetism), and $A_{AFM} = 0$ (no coupling) with $\Delta t = 1\;\textrm{fs}$ and $T = 20\;\textrm{ps}$.}
  \label{energy}
\end{figure}

\subsection{N\'eel wall structure in an antiferromagnet}\label{subsec:wall}

Consider an initial state with a N\'eel wall profile as in \cref{initial_phase_wall}. When the antiferromagnetic exchange parameter
$A_{AFM} = -3.0\times10^{-12}$, and $\mathbf{m}_A = \mathbf{m}_B$, it is expected that the system relaxes to a N\'eel wall structure
of ferromagnetic type; see \cref{{FMwallstable}}. When $A_{AFM} = 3.0\times10^{-12}$, the system relaxes to a N\'eel wall structure
of antiferromagnetic type; see \cref{FMwallstable_SchemeB}.

\begin{figure}[htbp]
    \centering
    \subfloat[Initial state]{\label{initial_phase_wall}\includegraphics[width=2.5in]{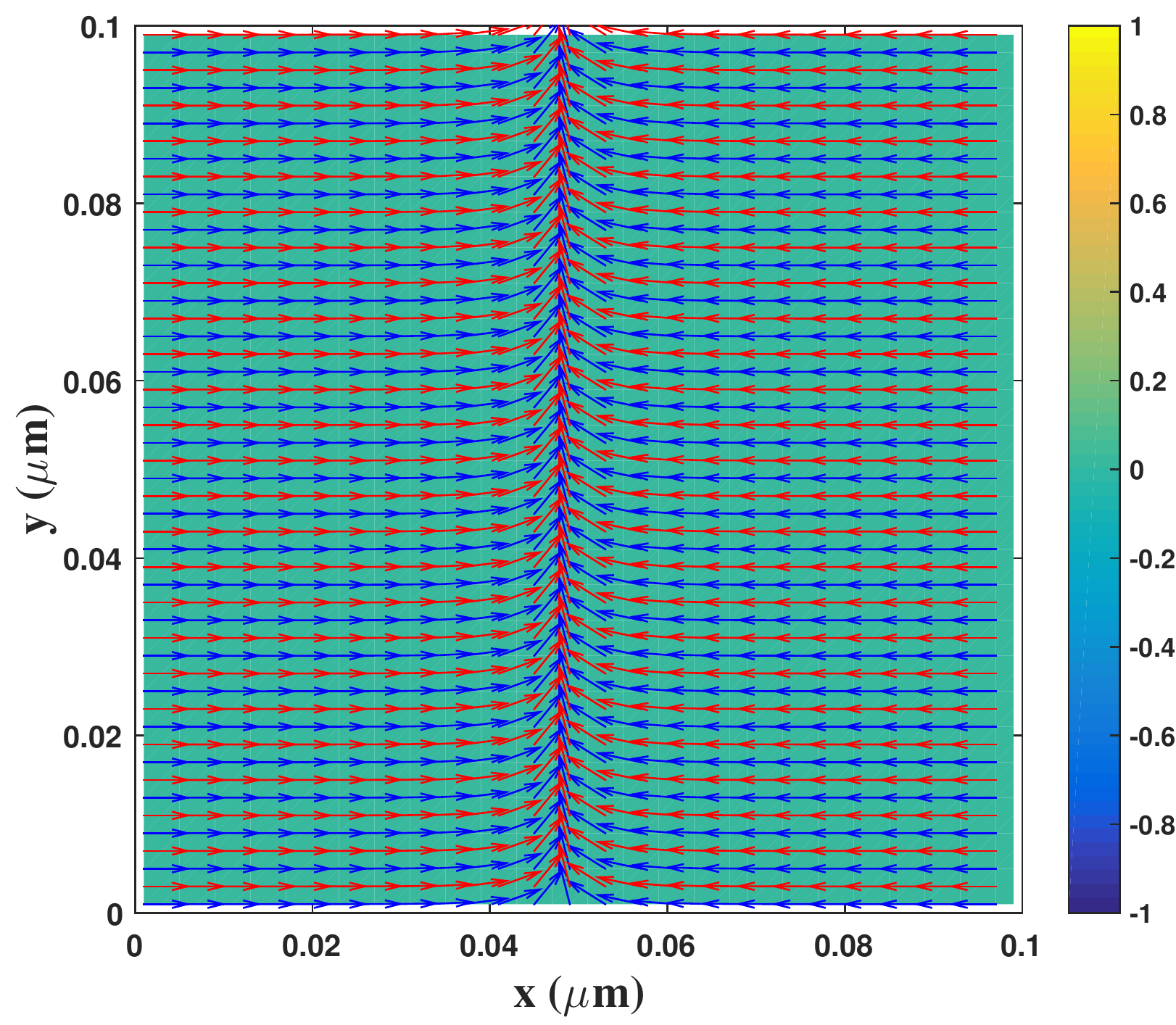}}
    \subfloat[Stable Ferromagnetic state]{\includegraphics[width=2.5in]{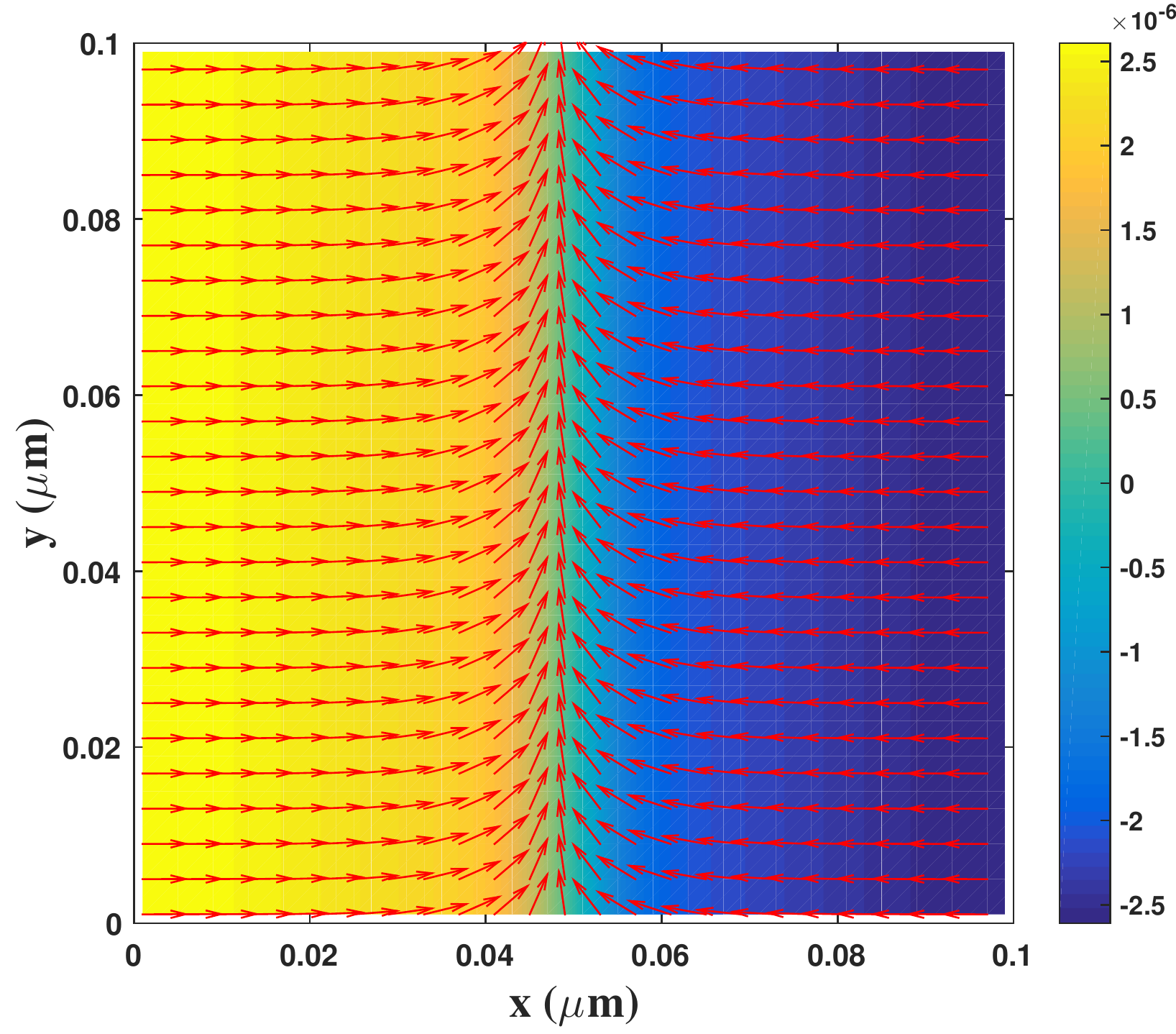}}
    \quad
    \subfloat[Stable state of sublattice A]{\includegraphics[width=2.5in]{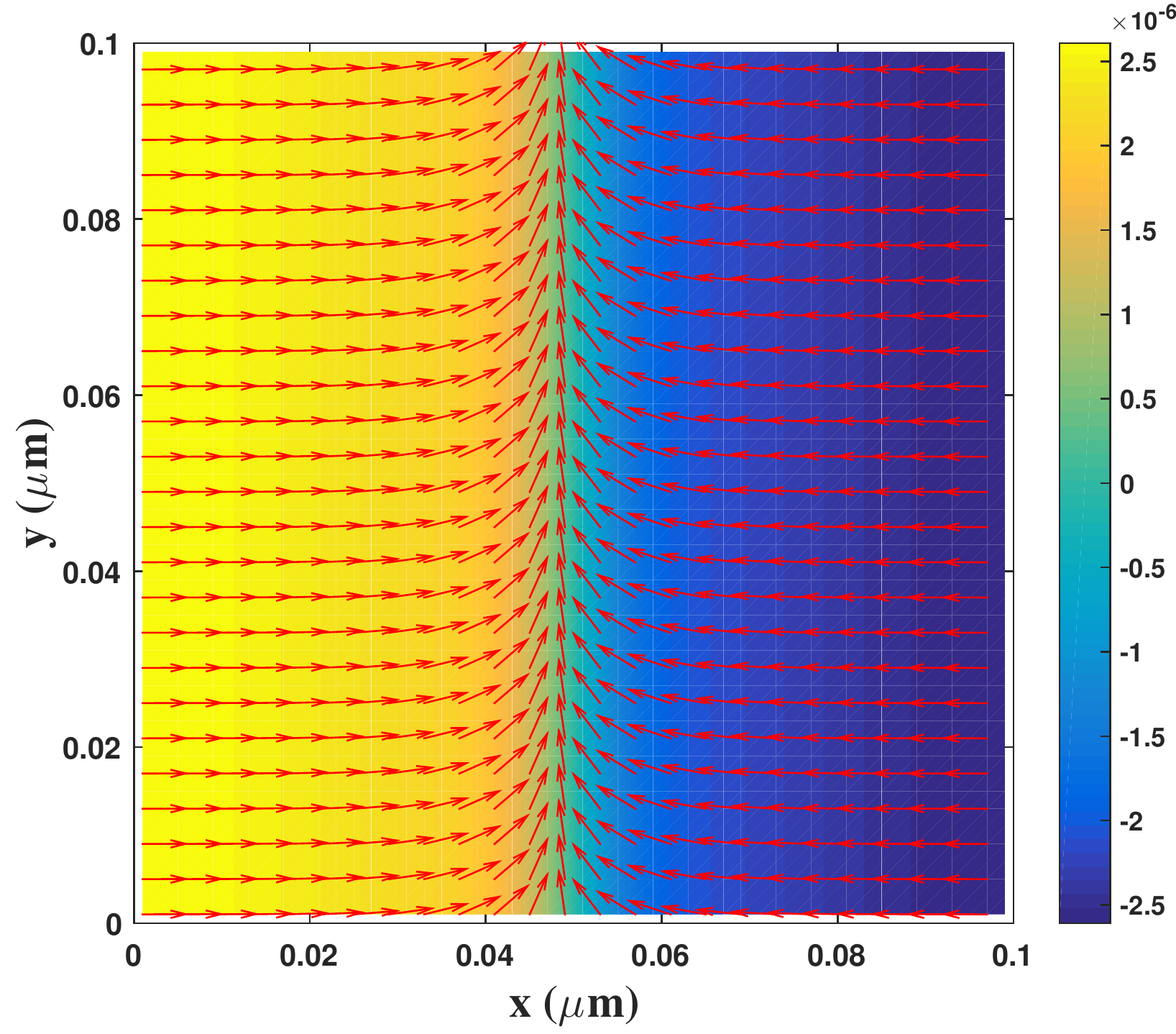}}
    \subfloat[Stable state of sublattice B]{\includegraphics[width=2.5in]{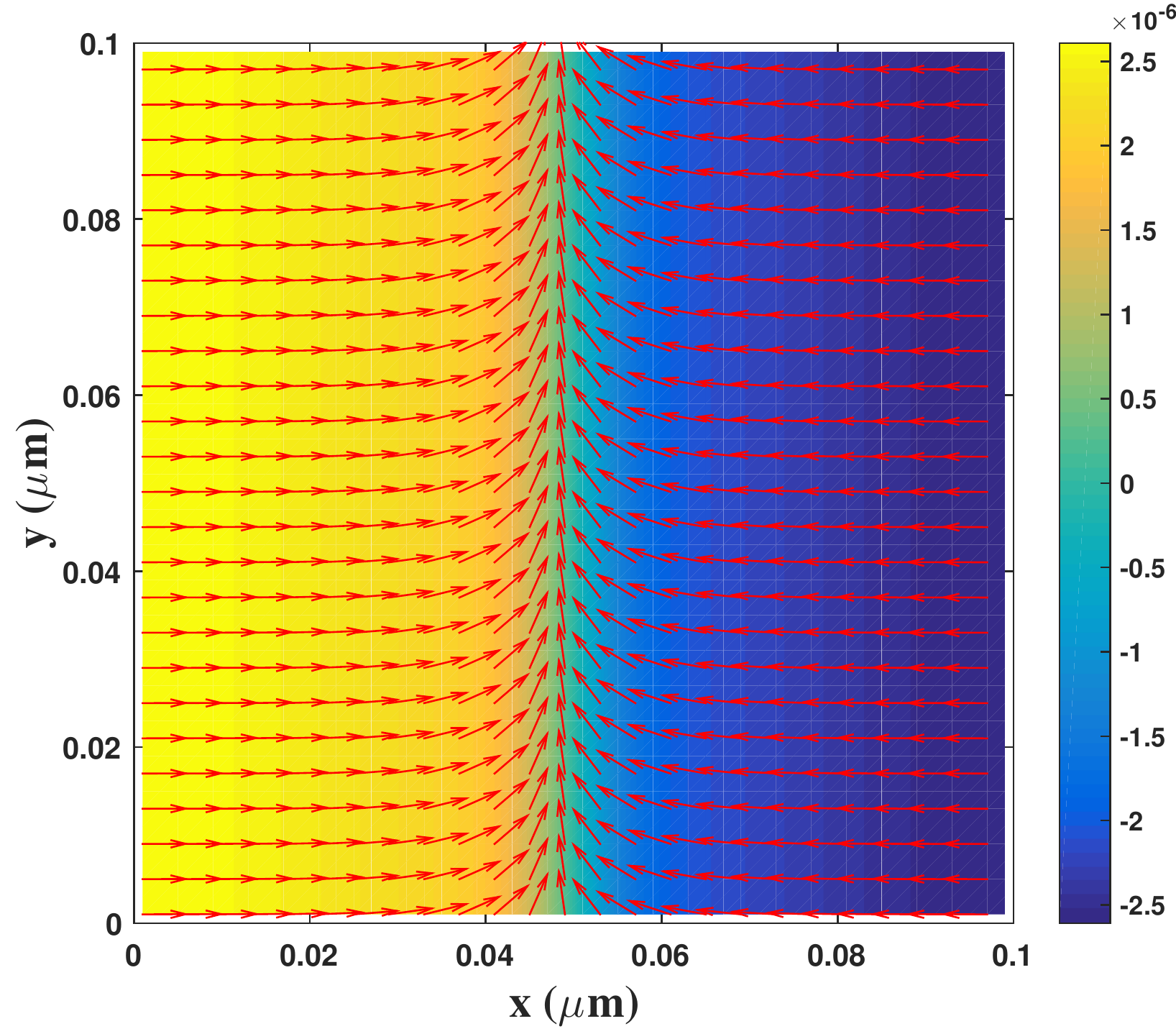}}
    \caption{Magnetization profile of $\mathbf{m} = (\mathbf{m}_A + \mathbf{m}_B)/2$ in the centered slice along the $z$ direction.
    Arrows are plotted by the first two components of the magnetization and the color is plotted by the third component.
    $\Delta t = 1\;\textrm{fs}$ and $T=2\;\textrm{ns}$. (a) Initial state; (b) Stable ferromagnetic state; (c) Stable state of sublattice A;
    (d) Stable state of sublattice B.}
    \label{FMwallstable}
\end{figure}

\begin{figure}[htbp]
    \centering
    \subfloat[$\mathbf{l} = \frac{\mathbf{m}_A - \mathbf{m}_B}{2}$]{\includegraphics[width=2.5in]{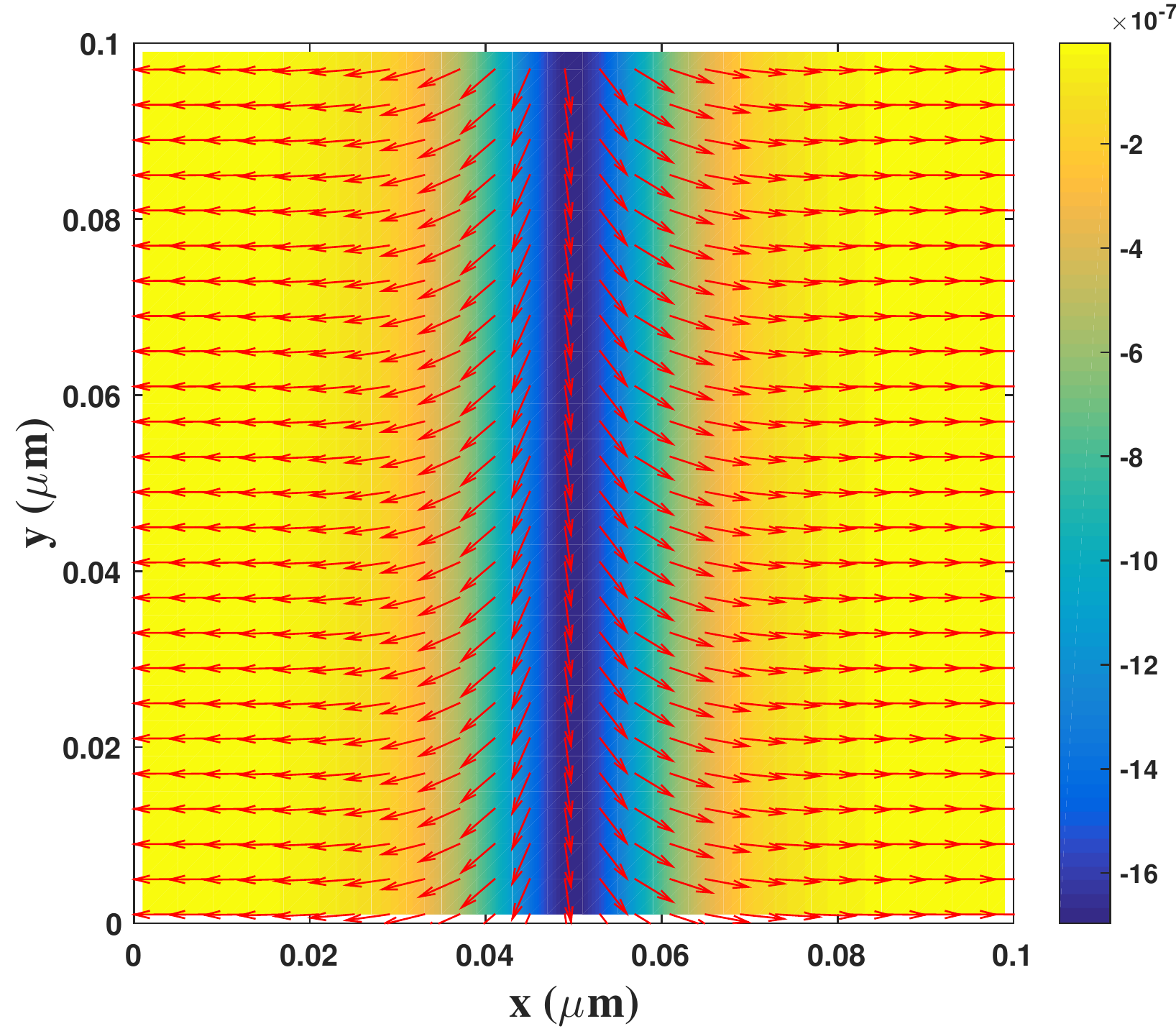}}
    \subfloat[$\mathbf{m} = \frac{\mathbf{m}_A + \mathbf{m}_B}{2}$]{\includegraphics[width=2.5in]{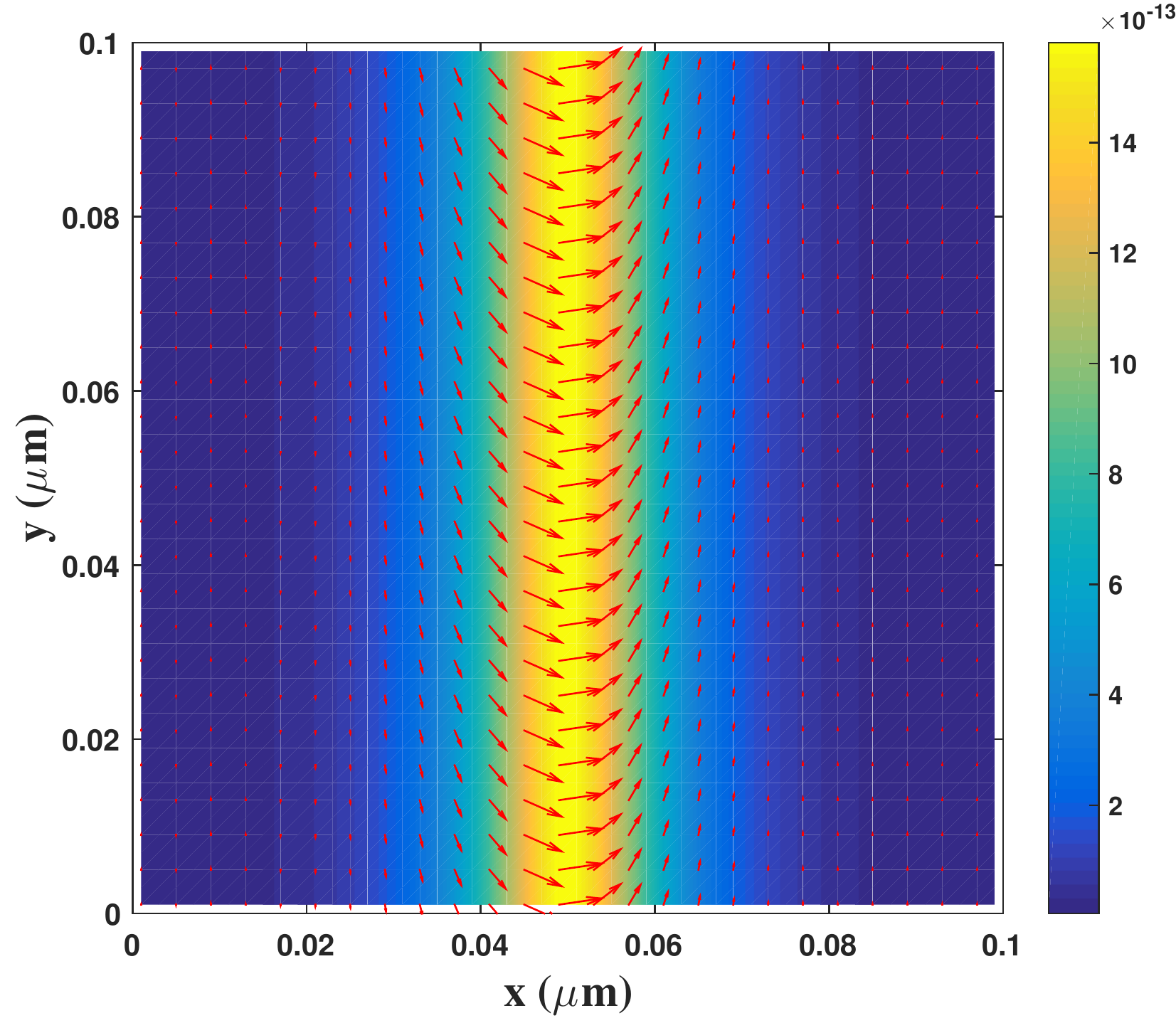}}
    \quad
    \subfloat[Stable state of sublattice A]{\includegraphics[width=2.5in]{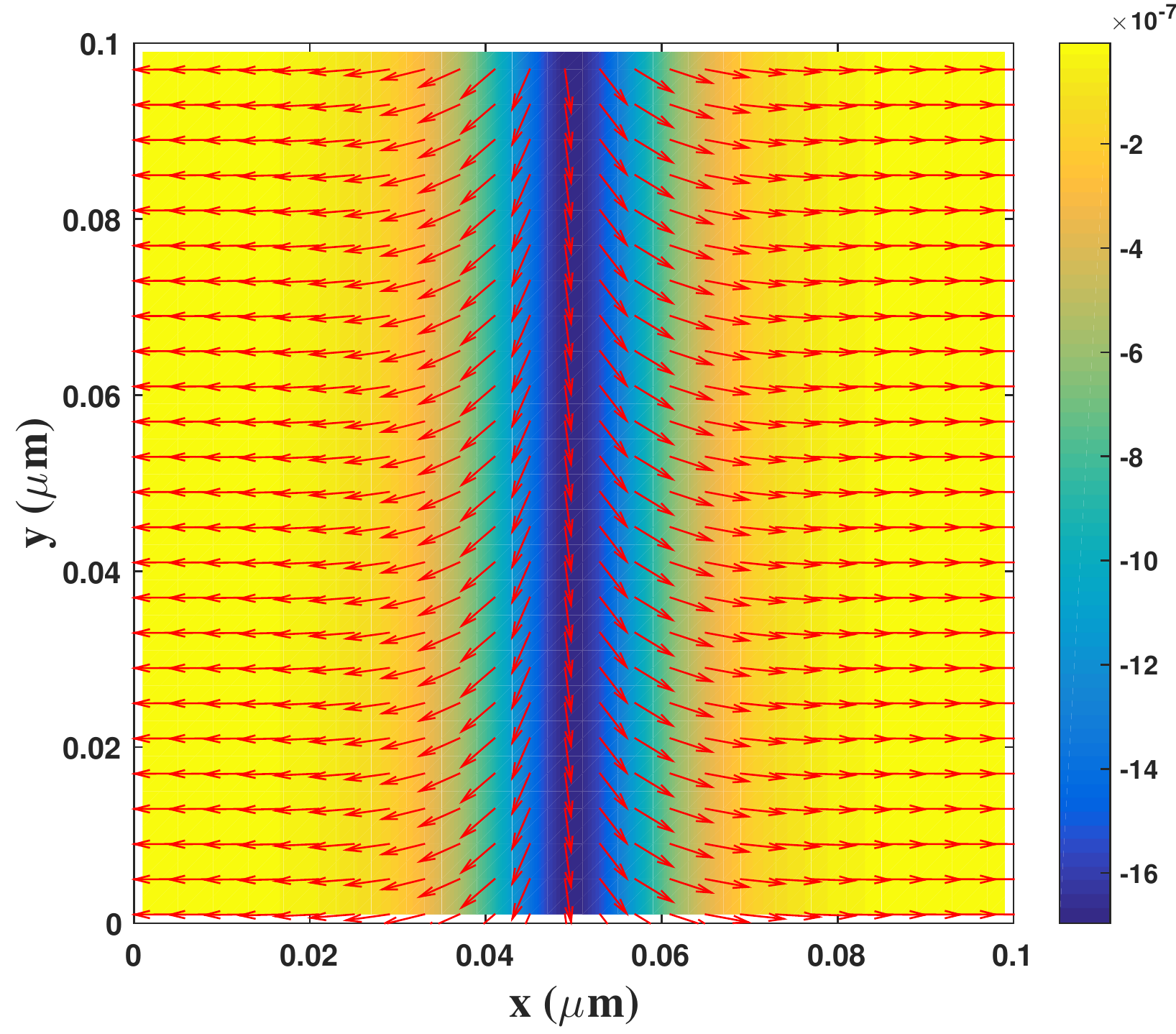}}
    \subfloat[Stable state of sublattice B]{\includegraphics[width=2.5in]{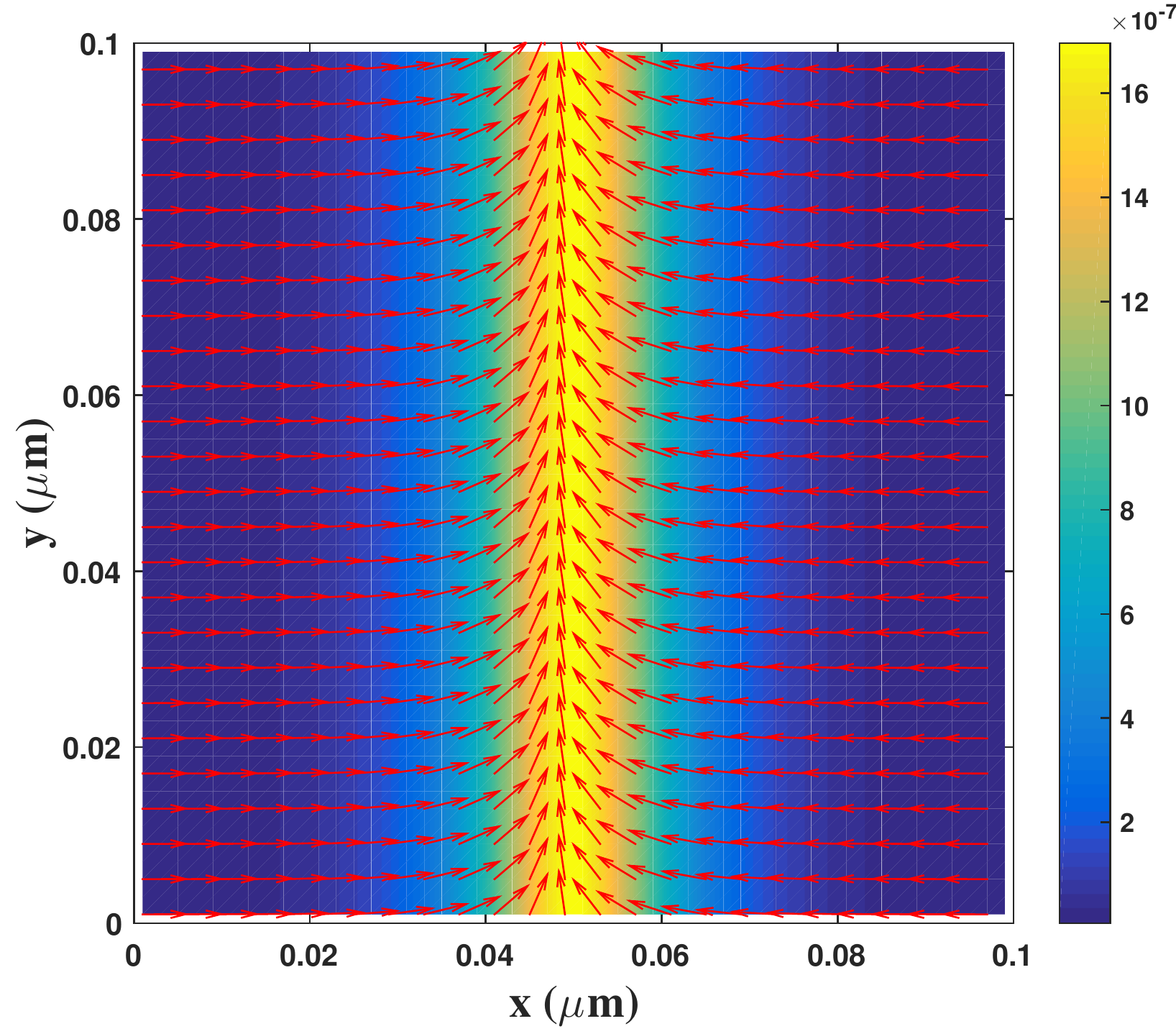}}
    \caption{Magnetization profile of the antiferromagnetic state in the centered slice along the $z$ direction.
    	Arrows are plotted by the first two components of the magnetization and the color is plotted by the third component.
    	$\Delta t = 1\;\textrm{fs}$ and $T=1\;\textrm{ns}$. (a) $\mathbf{l} = (\mathbf{m}_A - \mathbf{m}_B)/2$ at $1\textrm{ns}$;
    	(b) $\mathbf{m} = (\mathbf{m}_A + \mathbf{m}_B)/2$; (c) Sstable state of sublattice A. (d) Stable state of sublattice B.}
    \label{FMwallstable_SchemeB}
\end{figure}

\subsection{Phase diagram under the external field}\label{subsec:diagram}

It is known that AFMs are robust against magnetic perturbation, therefore the phase diagram requires stronger external fields.
The external field can be applied in three different ways. First, when an external field is applied perpendicular to the
easy-axis direction, magnetic moments of two sublattices and hence the magnetization $\mathbf{m} = (\mathbf{m}_A+\mathbf{m}_B)/2$
tend to align up with the field gradually until the net magnetization saturated. The external field is parallel to the easy-axis direction in the second and third cases. When the anisotropy energy is much
smaller than the antiferromagnetic exchange energy, the phase transition is known as the spin flop transition. When the
anisotropy energy is much larger than the antiferromagnetic exchange energy, magnetization jumps from zero to the saturation
magnetization directly, which is known as the spin flip transition \cite{Baltz2018Review,Coey2014Photocopy}.

In \cref{table:physicalparameter}, the anisotropy energy is much smaller that the antiferromagnetic energy, and thus
the spin flop transition happens. To observe the spin flip transition, we reduce the antiferromagnetic exchange parameter
$A_{AFM}$ to $3.0\times10^{-15}\;\mathrm{J}/\mathrm{m}$.
The whole simulation is started with an initial state $\mathbf{m}_A = -\mathbf{m}_B$ along the $x\textrm{-}axis$
and the external applied field $\mathbf{H}_\textrm{ext}=0\;\mathrm{T}\thicksim300\;\mathrm{T}$.
For both the $\mathbf{H}_{\bot}$ and the $\mathbf{H}_{||}$ phase transitions, an external field with $0\;\mathrm{T}$ is applied
and the system relaxes to a stable state. Afterwards, we increase the external field by a certain amount $\Delta\mathbf{H}_\textrm{ext}$
and let the system relax to an stable state. The external field is successively increased until $\mathbf{H}_\textrm{ext}=300\;\mathrm{T}$.
Averaged net magnetization is recorded as a function of $\mathbf{H}_\textrm{ext}$. The stopping criterion for a steady state is that the
relative change of the total energy is less than $10^{-9}$. Results shown in \Cref{phase_diagram} are in qualitative
agreements with the experimental results of MnF$_2$ \cite{Jacobs1961spinflop}.

\begin{figure}[htbp]
  \centering
  \includegraphics[width=5.5in]{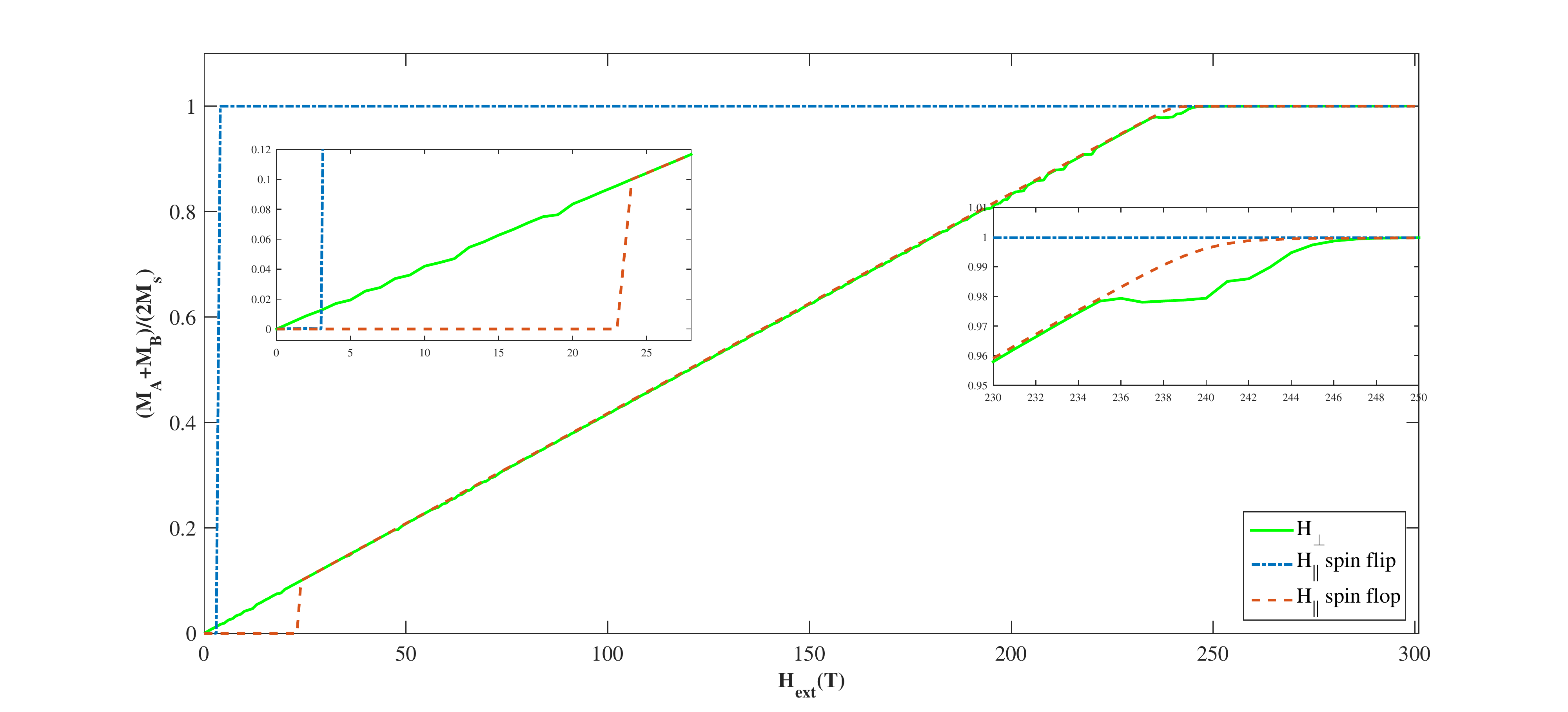}
  \caption{Phase diagram under the external field with parameters in \cref{table:physicalparameter} and $A_{AFM} = 3.0\times10^{-15}$
  	in the case of the spin-flip magnetic field.}\label{phase_diagram}
\end{figure}

\section{Conclusion}\label{sec:conclusion}

In this paper, we develop three Gauss-Seidel projection methods for the coupled system of Landau-Lifshitz-Gilbert equations with
an antiferromagnetic exchange coupling. Performance of these methods is verified in terms of accuracy and stability. In addition,
femtosecond magnetization dynamics, N\'eel wall structure, and phase transition under the external magnetic field in antiferromagnets
are studied. It is interesting that phase transitions observed in the simulation are qualitatively consistent
with the experimental results of MnF$_2$. In addition to experiments, the proposed methods open up an alternative way to understand
femtosecond magnetization dynamics in antiferromagnetic and ferrimagnetic materials.

\section*{Acknowledgments}
This work is supported in part by the grants NSFC 21602149 and 11971021 (J.~Chen), NSFC 11501399 (R.~Du), the Hong Kong Research Grants Council (GRF grants 16302715, 16324416, 16303318 and NSFC-RGC joint research grant N-HKUST620/15) (X.-P.~Wang).

\bibliographystyle{model1-num-names}
\bibliography{refs}

\end{document}